\PassOptionsToPackage{verbose=silent,expansion=true}{microtype}
\newcommand{\submissionmode}{1}

\documentclass[sigconf,nonacm]{acmart}
\usepackage{popets}
\graphicspath{{./}{Sections/Figure/}}

\setcopyright{none}
\copyrightyear{2027}

\acmYear{2027}
\acmVolume{2027}
\acmNumber{X}
\acmDOI{XXXXXXX.XXXXXXX}
\acmISBN{}
\acmConference{Proceedings on Privacy Enhancing Technologies}
\usepackage{lipsum}
\usepackage{pgfplots}
\usepackage{amsfonts}
\usepackage{amssymb}
\usepackage{amsmath}
\usepackage{tikz}
\usepackage{graphicx}
\usepackage{mathtools}
\usepackage{amsthm}
\usepackage{bbm}
\usepackage{tabularx}
\usepackage{booktabs} 
\usepackage{blindtext}
\usepackage{algorithm}
\usepackage{algpseudocode}
\usepackage[utf8]{inputenc}
\usepackage{csquotes}
\usepackage{subcaption}
\usepackage{float}
\usepackage{pgfplots}
\usepackage{hyperref}
\usepackage{multirow}
\usepackage{makecell}
\usepackage{enumitem}
\usepackage{url}
\usepackage{arydshln}

\usepackage{pgfplots}
\pgfplotsset{compat=1.18}
\usepackage{pgfplotstable}     
\usepackage{xcolor}
\usetikzlibrary{patterns, positioning}      
\usetikzlibrary{calc,fit,arrows.meta}
\usetikzlibrary{backgrounds}
\usetikzlibrary{fadings}
\usetikzlibrary{shadows.blur}
\tikzfading[name=softwhite,
  inner color=transparent!0,
  outer color=transparent!100
]
\usepgfplotslibrary{groupplots}

\definecolor{myorange}{HTML}{F28E2B}
\definecolor{mygreen}{HTML}{59A14F}
\definecolor{myblue}{HTML}{4E79A7}

\ifcase\submissionmode
\or 
\or
\fi

\pgfplotsset{compat=1.18}

\usetikzlibrary{calc}
\usetikzlibrary{arrows.meta}
\usetikzlibrary{matrix}

\newcommand{\bb}[1]{\mathbb{#1}}  
\newcommand{\cc}[1]{\mathcal{#1}} 
\newcommand{\round}[1]{\left\lceil #1 \right\rfloor} 

\newcommand{\ct}{\mathsf{ct}}
\newcommand{\sk}{\mathsf{sk}}
\newcommand{\pk}{\mathsf{pk}}

\newcommand{\dec}{\mathsf{Dec}}
\newcommand{\enc}{\mathsf{Enc}}

\ifcase\submissionmode
    \newcommand{\dhnote}[1]{{}\textcolor{red}{\textsf{[DH Comment: {#1}]}}}
    \newcommand{\hnnote}[1]{{}\textcolor{blue}{\textsf{[HN Comment: {#1}]}}}
    \newcommand{\jhnote}[1]{{}\textcolor{teal}{\textsf{[JH Comment: {#1}]}}}
\or
    \newcommand{\dhnote}[1]{} 
    \newcommand{\hnnote}[1]{}
    \newcommand{\jhnote}[1]{}
\or
    \newcommand{\dhnote}[1]{} 
    \newcommand{\hnnote}[1]{}
    \newcommand{\jhnote}[1]{}
\fi

\usepackage{amsthm}

\theoremstyle{plain}
\newtheorem{theorem}{Theorem}[section]

\newtheorem{lemma}[theorem]{Lemma}
\newtheorem{corollary}[theorem]{Corollary}
\theoremstyle{definition}
\newtheorem{definition}[theorem]{Definition}

\theoremstyle{remark}
\newtheorem{remark}[theorem]{Remark}

\usepackage[disable,textsize=tiny]{todonotes}
\begin{document}

\title[Private Embedding Lookup with Encrypted Compact Queries under FHE]{%
Private Embedding Lookup with Encrypted Compact Queries under Fully Homomorphic Encryption}

\author{Jung Hee Cheon}
\authornote{Authors are listed in alphabetical order.}
\affiliation{%
  \institution{Department of Mathematical Sciences, Seoul National University}
  \city{Seoul}
  \country{Republic of Korea}}
\affiliation{%
  \institution{CryptoLab Inc.}
  \city{Seoul}
  \country{Republic of Korea}}
\email{jhcheon@snu.ac.kr}

\author{Daehyun Jang}
\affiliation{%
  \institution{Department of Mathematical Sciences, Seoul National University}
  \city{Seoul}
  \country{Republic of Korea}}
\email{jadh0309@snu.ac.kr}

\author{Jaehee Kang}
\affiliation{%
  \institution{Department of Mathematical Sciences, Seoul National University}
  \city{Seoul}
  \country{Republic of Korea}}
\email{jaehee.kang422@snu.ac.kr}

\author{Hanee Rhee}
\affiliation{%
  \institution{Department of Mathematical Sciences, Seoul National University}
  \city{Seoul}
  \country{Republic of Korea}}
\email{hnrhee11@snu.ac.kr}

\renewcommand{\shortauthors}{Cheon, Jang, Kang, and Rhee}

\begin{abstract}
Many Natural Language Processing (NLP) and recommendation models begin by mapping discrete client inputs to embedding vectors.
Because these inputs can reveal sensitive information, this embedding step need to be protected in privacy-preserving inference.
Fully Homomorphic Encryption (FHE) enables inference over encrypted client data, but under FHE the embedding lookup is no longer a simple table access and has to be realized through homomorphic computation.
To keep the embedding table on the server and avoid transmitting encrypted embedding vectors from the client, we focus on the server-side lookup: the client sends only a small encrypted index.
Prior work from ICML 2024~\cite{KPLC24} first builds a one-hot vector from the encrypted index before multiplying it with the embedding table, and this one-hot generation is the dominant cost.
However, one-hot-based methods are expensive in FHE because they construct a $p$-dimensional selection vector by performing
an equality test for each coordinate, requiring $O(p\log p )$ homomorphic operations in total.

Our key observation is that private embedding lookup only requires a linearly independent representation of the encrypted index, not the one-hot basis itself.
Building on this observation, we propose Independent Vector Evaluation (IVE).
Instead of constructing a one-hot vector, IVE evaluates a linearly independent
vector built from successive powers of a single encrypted value, reducing the
vector-generation cost to $O(p)$ homomorphic operations.
It then recovers the same embedding vector through a precomputed change of
basis, which we instantiate with an orthogonal Discrete Cosine Transform (DCT)
to mitigate error amplification.

Our implementation demonstrates that this design improves amortized lookup time
by up to $78.4\times$ over the prior method.
We further evaluate its impact through end-to-end encrypted FastText inference, where embedding lookup is a major cost in the shallow model.
Concretely, on the Enron-Spam email-spam dataset, replacing one-hot generation with IVE reduces the share of vector generation in encrypted inference time from $99.6\%$ to $66.3\%$.

\end{abstract}



\maketitle

\section{Introduction}

Modern machine learning (ML) services increasingly run through an outsourcing model:
a client sends an input to a server, and the server evaluates a model on that input.
This model is natural in ML-as-a-Service (MLaaS), where clients such as smartphones,
wearable devices, or small applications may not have the memory or compute resources
to run the full model locally.
This delegated computation creates an immediate privacy problem.
The input sent to the server can reveal medical symptoms, financial behavior,
search intent, or the exact words typed by the user.
Thus, the client may want to use the model while keeping the input hidden from
the server that evaluates it.
Privacy-preserving machine learning (PPML) addresses this problem by allowing the
server to compute on the client's data without directly observing it.

Fully Homomorphic Encryption (FHE) is a natural tool for PPML because it supports
arithmetic computation directly on ciphertexts~\cite{Gen09,FV,BGV,CGGI,DM,CKKS,BRA}.
Compared with MPC or OT-based protocols, FHE keeps the online interaction simple:
the client encrypts its input, the server evaluates the circuit, and the client
or a later component decrypts only the intended result.
This is attractive in the MLaaS setting because the client does not need to
participate in every layer of the computation after sending the encrypted input.
Compared with differential privacy, FHE targets cryptographic confidentiality of
the input rather than hiding an individual's contribution through noise.
Among various homomorphic encryption schemes, the CKKS scheme~\cite{CKKS} is widely used for PPML because it efficiently supports approximate computation over complex numbers.

Recent progress in such FHE schemes, together with
FHE algorithms for neural-network inference~\cite{pmlr-v48-gilad-bachrach16,
dathathri2019chet,dathathri2020eva,chen2022thex,THOR,EFLLM}, has made encrypted
inference a realistic target for a growing range of ML workloads.
For arithmetic layers, this progress is especially visible: homomorphic linear
layers can use algorithmic reductions to plaintext matrix
multiplication~\cite{PCMM,CCMM,BatchMM}, while nonlinear functions such as
activation functions are increasingly optimized through polynomial
approximations~\cite{SOFTMAX}.
This paper focuses on a different bottleneck in such pipelines: privately
performing the embedding lookup for discrete user inputs.

For models whose inputs include discrete tokens or indices, inference begins with
an embedding layer.
Given an index $j\in[P]$, the embedding layer retrieves the vector
$M_j\in\mathbb{R}^d$, the column indexed by $j$ of an embedding table
$M\in\mathbb{R}^{d\times P}$, and passes this vector to the remaining layers.
Equivalently, if ${\bf e}_j$ is the one-hot vector for index $j$, then the
lookup can be written as $M_j=M{\bf e}_j$.
The embedding layer therefore converts a discrete index into the continuous
vector representation expected by the rest of the model.
This step is standard in word embeddings~\cite{Mikolov2013Efficient,Glove},
transformers~\cite{Vaswani2017Transformer,Devlin2019BERT}, and recommendation
models~\cite{Naumov2019DLRM}.
The privacy of the index itself is important.
A token index is not an arbitrary number: it identifies a word, subword, item, or
query chosen by the user.
Revealing whether a user typed a medical term, a financial term, or a politically
sensitive term may already reveal the sensitive part of the interaction.
Thus, for FHE-based inference over discrete inputs, the lookup from the index to
the embedding vector must also be handled without revealing the index.

There are two direct ways to handle this step.
The first way is client-side embedding: the client obtains the embedding table,
performs the lookup locally, encrypts $M_j$, and sends the encrypted vector to
the server.
This avoids encrypted lookup on the server, but it requires the server to
distribute the embedding table to the client and the client to send an encrypted
$d$-dimensional vector per token.
For GloVe.42B.300d, the table has $1.9\mathrm{M}$ vectors of dimension $300$, which
is about $2.3$GB in single precision before any encryption is considered.
Thus, the client must receive and store a large table.
If the table is proprietary, distributing it also exposes model parameters to the
client.
Even when the table is public, storage, redistribution, and update costs remain
substantial.
The encrypted vector sent per token also adds substantial client-to-server communication, since the transmitted ciphertexts scale with the embedding dimension $d$. This cost becomes more pronounced for multi-token inputs or large embedding dimensions. These requirements are unattractive for lightweight clients.

The other way is server-side embedding: the client sends an encrypted query
derived from the index $j$, and the server performs the lookup on ciphertexts.
This keeps the table on the server and avoids sending the embedding vector from
the client, but it shifts the table-size dependent lookup cost to the server.
We call this server-side task \emph{Private Embedding Lookup (PEL)}: the client sends
an encrypted query derived from a private index, and the server produces
ciphertexts encoding the corresponding embedding vector in a form suitable for
the following homomorphic layers.

The main difficulty in server-side embedding is that private lookup no longer
behaves like ordinary memory access.
In plaintext, the lookup is just a memory access.
Under FHE, however, the server cannot branch on an encrypted index or use it as a
memory address.
The lookup must be replaced by arithmetic operations over ciphertexts, and the
cost grows with the table size $P$.
For example, lookup-table methods may evaluate an interpolating
polynomial~\cite{LUT2}, while embedding approaches often reduce the lookup to a
matrix-vector multiplication.
As a result, a step that is essentially a memory access in plaintext becomes a
table-size-dependent ciphertext computation under FHE.
This change is especially costly for embedding layers because their tables are often large: vocabulary sizes such as BERT's roughly $30\mathrm{K}$ tokens, GPT-2's roughly $50\mathrm{K}$ tokens, or GloVe.42B.300d's $1.9\mathrm{M}$ tokens. As a result, the embedding lookup stage can already become a dominant component of the encrypted inference cost, even before the model reaches the usual linear and nonlinear layers.

Efficient PEL is therefore needed to make server-side embedding a practical
alternative to client-side embedding: the goal is to keep the embedding table on
the server and avoid client-to-server transmission of encrypted embedding
vectors, without turning the private lookup itself into the dominant cost.
We focus on this server-side PEL setting in the rest of this section.
The remaining question is how existing protocols handle the resulting private
lookup problem.

\subsection{Prior Work Comparison}
\label{subsec:intro-design-space}

We compare how existing protocols handle this private lookup problem.
For efficiency, the works considered in~\cite{KPLC24,HE-LRM} compress the table
before private lookup is applied.\label{subsubsec:table-compression}
Instead of using one large table $M\in\mathbb{R}^{d\times P}$, table compression
represents it by $\ell$ smaller subtables
$M^{(1)},\ldots,M^{(\ell)}\in\mathbb{R}^{d\times p}$.
An original index $j\in[P]$ is mapped to a tuple of sub-indices
$(j_1,\ldots,j_\ell)\in[p]^\ell$, and the embedding vector is reconstructed by
retrieving one vector from each subtable and summing them:
\[
    M_j \approx \sum_{t=1}^{\ell} M^{(t)}_{j_t}.
\]
Thus a lookup over a table of size $P$ is replaced by $\ell$ lookups over tables
of size $p$.
We treat this compression step as a front-end choice: once it produces subtables
of size $p$, our focus is the cryptographic cost of each private subtable
lookup.

\begin{table*}[t!]
\centering
\scriptsize
\setlength{\tabcolsep}{3pt}
\renewcommand{\arraystretch}{1.16}
\caption{Design-space comparison for private embedding lookup. For the FHE-only
rows, table compression gives $\ell$ private lookups over subtables of size $p$,
followed by summation of the resulting vectors. FastQuery is included as an
adjacent 2PC+HE comparison point rather than as a protocol normalized under this
FHE-only compressed-table metric. Query size counts the plaintext query
representation before encryption and packing. The last row shows the design
point pursued in this work.}
\label{tab:intro-related-comparison}
\begin{tabularx}{0.98\textwidth}{@{}>{\centering\arraybackslash}p{0.09\textwidth}>{\centering\arraybackslash}p{0.07\textwidth}>{\centering\arraybackslash}p{0.18\textwidth}>{\centering\arraybackslash}p{0.115\textwidth}>{\centering\arraybackslash}p{0.115\textwidth}>{\centering\arraybackslash}p{0.135\textwidth}>{\raggedright\arraybackslash}X@{}}
\toprule
\multirow{2}{*}{Work} & \multirow{2}{*}{Setting/output}
& \multicolumn{2}{c}{Client query}
& \multicolumn{2}{c}{Server vector recovery}
& \multirow{2}{*}{Takeaway} \\
\cmidrule(lr){3-4}\cmidrule(lr){5-6}
& & Representation & Size & One-hot vector on server? & Cost before table multiplication & \\
\midrule
FastQuery~\cite{FastQuery}
& 2PC+HE
& Masked one-hot-vector query with offline randomness; not based on the
table-compression front-end.
& Not directly comparable
& No
& Not an FHE-only encrypted-output pipeline
& Adjacent setting: output is share-based, so the FHE-only encrypted-output
metric below does not directly apply. \\
\midrule
HE-LRM~\cite{HE-LRM}
& FHE
& Encrypted one-hot vectors $\mathsf{Enc}({\bf e}_{j_t})$.
& \textbf{Large}: $\Theta(\ell p)$ vector entries
& No; client provides vectors
& None before table multiplication
& Saves server recovery work, but pays for it with large client communication. \\
\midrule
Kim et al.~\cite{KPLC24}
& FHE
& Encrypted scalar sub-indices $\mathsf{Enc}(j_t)$.
& \textbf{Compact}: $\Theta(\ell)$ scalars
& Yes
& \textbf{High}: $\Theta(\ell p\log p)$ operations
& Communication is compact, but homomorphic one-hot generation becomes the bottleneck. \\
\midrule
\textsc{IVE-PEL} (this work)
& FHE
& Encrypted scalar sub-indices $\mathsf{Enc}(j_t)$.
& \textbf{Compact}: $\Theta(\ell)$ scalars
& No
& \textbf{Lower}: $\Theta(\ell p)$ operations
& Hits the target point: compact query and no server-side one-hot generation. \\
\bottomrule
\end{tabularx}
\end{table*}


\paragraph{Communication Cost vs. Computation Cost.}
These protocols expose the tradeoff we aim to resolve: keeping the compact
communication of scalar-index queries without making the server homomorphically
generate a one-hot vector.
Existing works~\cite{FastQuery,HE-LRM,KPLC24} take different approaches to this private lookup problem.
FastQuery~\cite{FastQuery} targets private embedding lookup, but in a 2PC+HE protocol with a masked one-hot online query and additive shares setting, making it somewhat less aligned with the FHE-based inference setting that we consider.

For the FHE-only line of work, the important distinction is the tradeoff between client communication and server-side one-hot computation.
HE-LRM~\cite{HE-LRM} lets the client provide an encrypted one-hot vector $e_j\in \{0,1\}^p$.
Once the server has this vector, the lookup itself is a homomorphic linear operation with the table, $M_j=Me_j\in \bb R^d$.
This keeps the server-side lookup simple, but the client must communicate a length-$p$ encrypted vector for each private lookup, which leads to substantial blowup if $P$ is large.

The method of Kim et al. (ICML 2024)~\cite{KPLC24} instead allows the client to send only an encrypted scalar index $j\in [p]$.
This keeps the query compact.
The server, however, must first transform the encrypted index $j$ into the corresponding one-hot vector $e_j$ before applying the same linear operation.
Kim et al. introduced a function called the Encrypted Indicator Function (EIF) to perform this operation, but it is extremely costly, even compared to the subsequent multiplication with the table.

To sum up, if the client sends the one-hot vector, the server-side computation is simple but the communication is large.
If the client sends only the scalar index, the communication is small but the server additionally computes the one-hot vector homomorphically.
This one-hot generation step is homomorphically expensive and becomes the main server-side bottleneck in the lookup regimes we target.
Table~\ref{tab:intro-related-comparison} summarizes this communication-versus-one-hot computation tradeoff for the FHE-only approaches.
It lists FastQuery separately as an adjacent 2PC+HE comparison point.

\subsection{Our Method}
\label{subsec:intro-method}

Our goal is the missing design point identified above: keep the communication benefit of sending encrypted scalar sub-indices while avoiding expensive server-side one-hot generation.
In other words, we aim to retain a compact query representation while achieving the effect of one-hot generation with a computational cost that is not dominant compared to the subsequent table multiplication.

We achieve this by avoiding one-hot generation altogether.
The server does not actually need the one-hot vector itself.
It only needs to produce the same selected embedding vector.
This suggests an algebraic alternative.
The server only needs an encrypted representation that uniquely identifies the
index and can be converted, by a plaintext precomputed transformation, into the
same selected embedding.
Instead of representing the index $j$ by the one-hot vector ${\bf e}_j$, we
associate it with another vector ${\bf v}_j$.
As long as the vectors $\{{\bf v}_j\}$ are linearly independent, there exists a
change-of-basis matrix $L$ satisfying $L{\bf v}_j={\bf e}_j$.
Then the lookup result is unchanged:
\[
    M{\bf e}_j = ML{\bf v}_j.
\]
The matrix $ML$ is plaintext and can be precomputed by the server.
Therefore, the online task is no longer constructing the one-hot vector but
becomes evaluating a suitable independent vector ${\bf v}_j$.

This changes the design space.
One-hot vectors are only one possible basis, and they are not friendly to
homomorphic evaluation from an encrypted index.
Once we allow any linearly independent vector family, we can choose vectors that
are much easier to generate homomorphically and still recover the same embedding
after the precomputed linear transformation.
Our method follows this principle through Independent Vector Evaluation (IVE).
The challenge is not merely to find another basis, but to choose one that is
cheap to generate under FHE and numerically stable under CKKS approximation.
This imposes two additional
requirements: efficient homomorphic generation and numerical stability.

\paragraph{Efficiency Requirement.}
The vector family should be easy to generate under FHE.
This suggests vectors generated from powers of one value, such as
\[
    {\bf v}(x)=(1,x,x^2,\ldots,x^{p-1})^\top.
\]
Once the server has an encryption of $x$, it can generate powers of $x$ using
ordinary arithmetic operations.
This is much more arithmetic-friendly than testing, coordinate by coordinate,
which position equals the hidden index.

\paragraph{Stability Requirement.}
Efficiency alone is not enough.
A poorly conditioned change of basis can amplify
the error in the evaluated vector.
We therefore want the matrix $L$ formed by the vectors ${\bf v}_j$ to be orthogonal,
or at least very well conditioned.
Together with the power-vector form, this points to roots of unity: in the
complex setting, powers of roots of unity give the Fourier basis, which is
unitary. Thus the guiding principle is to use powers for efficient homomorphic generation and roots of unity for numerical stability.

\paragraph{Real-valued Version.}
Further, since embedding tables are real-valued, so we use a real-valued version of this idea.
We take a shifted root-of-unity construction and separate the real and imaginary
parts of the resulting powers.
This realification yields
the Discrete Cosine Transform (DCT) matrix used in our protocol.
The DCT choice keeps the change of basis orthogonal while letting the server generate the vector entries from powers of a single encrypted value.
This gives $O(p)$ homomorphic vector generation for a table of size $p$.


\subsection{Our Contributions}
\label{subsec:intro-contributions}

This work studies Private Embedding Lookup (PEL) for FHE-based inference over
discrete inputs for the feasibility for lightweight client situation.
We aim to design PEL to support server-side lookup with compact client queries while keeping the homomorphic computation small enough for end-to-end inference.
Our contributions are as follows.
\begin{itemize}
 \item \textbf{Independent Vector Evaluation.}
    We propose Independent Vector Evaluation (IVE), which replaces encrypted
    one-hot generation with the evaluation of an arithmetic-friendly linearly
    independent vector ${\bf v}_j$.
    The key observation is that one-hot vectors are only one possible basis:
    if $L{\bf v}_j={\bf e}_j$, then the server can precompute $ML$ and recover
    the same embedding as $ML{\bf v}_j=M_j$.
    We instantiate this idea using a DCT-based orthogonal construction, which
    supports power-based homomorphic generation while controlling numerical
    error amplification in CKKS.
    For a size-$p$ lookup, IVE reduces vector generation from $O(p\log p)$ to
    $O(p)$ homomorphic operations.

    \item \textbf{An Efficient PEL Pipeline via IVE and PCMM.}
    We integrate IVE with plaintext-ciphertext matrix multiplication
    (PCMM)~\cite{PCMM} to build an efficient PEL pipeline for multiple token
    queries.
    IVE reduces the cost of encrypted vector generation, while PCMM reduces the
    linear cost of applying the embedding table to many encrypted query vectors
    in parallel.
    Under the same CKKS parameters and lookup size, our implementation shows
    that the proposed IVE-based PEL pipeline improves the amortized lookup time
    by $18.5$--$78.4\times$ over the one-hot-based method of Kim et al. at
    $\log p=10$, as reported in Table~\ref{tab:compare}.

    \item \textbf{End-to-end Encrypted FastText Inference.}
    We implement an end-to-end encrypted inference pipeline for a FastText-style text classification model on privacy-sensitive text datasets.
    The pipeline uses our PEL method as the encrypted embedding layer and then performs the remaining linear classification homomorphically.
    Compared with replacing the embedding layer by the one-hot-based lookup method of Kim et al., our method substantially reduces the amortized end-to-end inference time per input example.
    On the spam mail classification dataset (Enron-Spam), IVE-PEL reduces the share of vector generation in the encrypted inference time from $99.6\%$ under Kim et al.'s method to $66.3\%$, as reported in Table~\ref{tab:fasttext}.   
    These results show that our faster lookup method also improves the overall encrypted inference time.
\end{itemize}



\subsection{Related Work}
\label{subsec:intro-related}

The closest prior systems for private embedding lookup are HE-LRM~\cite{HE-LRM},
the method of Kim et al.~\cite{KPLC24}, and FastQuery~\cite{FastQuery}; these
are compared in Section~\ref{subsec:intro-design-space}.
We use this subsection only to clarify how PEL differs from adjacent primitives
such as PIR and Homomorphic Lookup.

\paragraph{Relation to PIR}
Private Information Retrieval (PIR)~\cite{firstPIR} also hides an index while
retrieving a database entry.
At a high level, this sounds very close to PEL.
The client has a private index, the server has a table, and the protocol should
retrieve the selected entry without revealing the index.
The closest PIR analogue is single-server, single-round, linear PIR, where the
client sends one encrypted query and the server combines table entries according
to that query.
The difference is the role of the output.
In PIR, the selected record is usually returned to the client.
The client can then decrypt it and may use preprocessing or postprocessing to
recover the final record; many PIR communication/computation trade-offs are built
around this final-client-output model~\cite{TradeoffsinPIR}.
In PEL, the embedding is not the final output of the application.
It is an intermediate ciphertext inside an FHE inference pipeline.
The next homomorphic layer must be able to consume it directly.
If the client had to decrypt, reconstruct, and re-encrypt the embedding, the FHE
pipeline would be interrupted.
Thus PEL needs a lookup protocol whose encrypted output is already in the right
form for the following homomorphic computation.

\paragraph{Relation to Homomorphic Lookup}
Homomorphic lookup algorithms~\cite{LUT1,LUT2,TFHELUT} evaluate
functions of encrypted inputs and are closely related to scalar-input lookup.
They are relevant because an embedding lookup can also be viewed as a function
of an encrypted index.
However, such work usually focuses on evaluating a function inside a
homomorphic circuit.
Client-to-server communication is not usually the main design choice in that
cost model.

PEL differs in two ways.
First, the query representation is part of the problem.
The client might send a one-hot query, a scalar-index query, or another compact
encoding, and this choice directly changes both communication and server-side
work.
Second, the output is vector-valued and must be encoded in a CKKS format suitable
for downstream linear layers.
Standard homomorphic lookup analyses do not directly capture this combination of
query communication, vector output, CKKS encoding, and compatibility with the
rest of the FHE inference pipeline.

\section{Preliminaries}
In this section, we present the notation and review the main tools used in our protocol.

For a positive integer $p$, we denote $[p]:=\{0,1,\ldots,p-1\}$.
Throughout, $\log$ denotes $\log_2$.
We use bold lower-case letters (e.g., $\mathbf{v}$) for vectors and upper-case letters (e.g., $M$) for matrices.
We write $\sqrt{-1}$ for the imaginary unit (reserving $i$ for indexing).
For a vector $\mathbf{v}$, we write $\mathbf{v}[i]$ for its $i$-th component.
For a matrix $M$, we write $M_j$ for its $j$-th column.

\subsection{Fully Homomorphic Encryption}
Fully Homomorphic Encryption (FHE) enables arithmetic directly over encrypted data~\cite{Gen09}.
We recall the basic definition of an FHE scheme.
A Fully Homomorphic Encryption scheme $\Pi_{\sf HE}$ is a tuple of probabilistic polynomial time (PPT) algorithms
\[
\Pi_{\sf HE}=({\sf KeyGen}, {\sf Enc}, {\sf Dec}, {\sf Eval}).
\]
Here, ${\cc {SK}}$ and ${\cc {PK}}$ denote secret key space and public key space, respectively.
We let $\pk$ include the public encryption key and the evaluation keys needed for homomorphic operations.
Also, ${\cc {M}}$ and ${\cc {C}}$ denote message space and ciphertext space, respectively.
\begin{itemize}
  \item ${\sf KeyGen}(1^\lambda)\to (\sk,\pk)$: Given security parameter $\lambda$, outputs a secret key $\sk\in{\cc{SK}}$ and public key $\pk\in {\cc {PK}}$.
  \item ${\enc}(\pk,m)\to \ct$: Given message $m\in \cc M$ and public key $\pk$, outputs a ciphertext $\ct\in \cc C$.
  \item ${\dec}(\sk, {\sf ct})\to m$: Given ciphertext $\ct \in \cc C$ and secret key $\sk$, outputs a message $m\in \cc M$.
  \item ${\sf Eval}(\pk, (\ct_i)_{i\in[n_1]};C)\to (\ct'_j)_{j\in[n_2]}$: Given public key $\pk$, an $n_1$-tuple of ciphertexts, and an $n_1$-input-$n_2$-output circuit $C$, outputs an $n_2$-tuple of ciphertexts.
\end{itemize}

\subsection{CKKS Scheme}
In this work, we use the CKKS scheme~\cite{CKKS}, which supports approximate (fixed-point) arithmetic over complex numbers.
We introduce the basic concepts required for our construction; additional details are in Appendix~\ref{Appendix:CKKS}.

\paragraph{CKKS Message and Ciphertext.}
Let $N$ be a power-of-two. We define $\cc R_N:=\mathbb{Z}[X]/(X^N+1)$.
For integers $q_0,\dots,q_L$,    
set $Q_\ell := \prod_{i=0}^{\ell} q_i$ for $\ell=0,\dots,L$.
Let $\cc R_{N,Q_{\ell}}:=\mathbb Z_{Q_{\ell}}[X]/(X^N+1)$.
We call $N$ the \textit{ring degree}, $\ell$ the \emph{level},  and $Q_{\ell}$ the ($\ell$-level) \textit{modulus}.
Then, the CKKS \emph{message space} $\cc M$ is $\bb C^{N/2}$.
The ($\ell$-level) CKKS \emph{ciphertext space} $\cc C_{\ell}$ is $\cc R_{N,Q_{\ell}}^2=\{(a,b):a,b\in \cc R_{N,Q_{\ell}}\}$.

\paragraph{CKKS Encryption.}
CKKS encryption $\enc$ proceeds as follows.
A complex vector ${\bf m}\in \bb C^{N/2}$ is first scaled by a \emph{scale factor} $\Delta$ and encoded (via an FFT-based map) into a polynomial $\mu \in \cc R_N$.
(Using a larger scale factor $\Delta$ generally yields higher-precision homomorphic computations.)
The encryption algorithm (to level $\ell$) outputs a ciphertext $(a,b)\in \cc C_\ell$ such that
$as+b=\mu+e_{\sf enc}\pmod{Q_\ell}$ for a secret $s={\sf sk}\in \cc R_N$ and small encryption error $e_{\sf enc}$.
Accordingly, decryption $\sf Dec$ takes a level-$\ell$ ciphertext $(a,b)$, computes $as+b \pmod{Q_\ell}$ with respect to the secret key $s$, and then decodes the result to recover a vector in $\bb C^{N/2}$.

\paragraph{CKKS Operations.}
For CKKS ciphertexts, homomorphic addition and multiplication are supported over $N/2$ complex values in a SIMD manner.
\begin{itemize}
\item {\sf Add}: Adds two ciphertexts or one ciphertext and a message, denoted as $\ct_1+\ct_2$ or $\ct+{m}$.
\item {\sf Mult}: Multiplies two ciphertexts or one ciphertext and a message, denoted as $\ct_1\cdot \ct_2$ or ${m}\cdot\ct$.
\item {\sf Conj}: Conjugates a ciphertext $\ct$ homomorphically, denoted as $\overline {\sf ct}$.
\end{itemize}

The cost of the CKKS operation depends on the ciphertext level: Operations on higher-level ciphertexts incur greater cost, approximately proportional to the level.
Each ciphertext-ciphertext multiplication, and typically each multiplication by an encoded non-integer plaintext, consumes one level of a CKKS ciphertext\footnote{As an exception, multiplying a ciphertext by an integer (including Gaussian integers in $\bb Z[\sqrt{-1}]$) does not consume a level.}. 
Once all levels are exhausted, further homomorphic multiplications are no longer supported. The total number of levels consumed by a homomorphic circuit is called its (multiplicative) \emph{depth}.
CKKS bootstrapping $\sf BTS$ refreshes a ciphertext by restoring its level while approximately preserving the underlying message~\cite{CKKSBTS}. 


Matrix multiplication is widely used in practice, and extensive prior work has studied how to perform it efficiently under FHE, including both encoding strategies and optimized multiplication procedures. Among these, the multiplication of a plaintext matrix with a ciphertext matrix is commonly referred to as $\sf PCMM$. A recent advance enables this operation to be implemented very efficiently~\cite{PCMM}.

\begin{itemize}
    \item ${\sf PCMM}_A({\sf CT})$: For a plaintext matrix $A\in\bb C^{r\times s}$ and ciphertext(s) ${\sf CT}$ encoding a matrix $B\in\bb C^{s\times t}$ under a prescribed layout,
    outputs ciphertext(s) encoding $AB$ under the corresponding output layout.
\end{itemize}

\section{PEL: Private Embedding Lookup}\label{sec:pel}
In this section, we formalize Private Embedding Lookup (PEL) and discuss its properties.

Let $\Pi_{\sf HE}=({\sf KeyGen},{\sf Enc},{\sf Dec},{\sf Eval})$ be a homomorphic encryption scheme with message space $\cc M$ and ciphertext space $\cc C$.
For positive integers $r,d$, let ${\sf VecEcd}: \bb R^d\to \cc M^r$ be a vector encoding map with corresponding decoding map ${\sf VecDcd}$.
Let $\mathcal T\subset\bb R^{d\times p}$ be a finite-precision family of embedding tables, and write $M_j:=M_{:,j}\in\bb R^d$ for the $j$-th column of $M$.
Let ${E}:\bb R^d\times \bb R^d\to [0,\infty]$ be an error metric function and let $u$ be the inverse-error target, so that the target error is $1/u$.
Then we define Private Embedding Lookup as follows:

\begin{definition}[Private Embedding Lookup]\label{def:pel}
        A \textit{Private Embedding Lookup} protocol $\Pi_{\sf PEL}$ is a tuple of PPT algorithms
        \[
        \Pi_{\sf PEL}=({\sf Setup},{\sf Query}, {\sf Emb})
        \]
        parameterized by $(\Pi_{\sf HE},{\sf VecEcd},\mathcal T,E, u)$, between a client $\sf C$ that privately holds the HE secret key $\sk\in \cc{SK}$ and an index $j\in[p]$, and a server $\sf S$ that privately holds the embedding table $M\in \mathcal T$.
        \begin{itemize}
            \item ${\sf Setup}(1^{\lambda})\to (\sk, {\sf pp})$: Given security parameter $\lambda$, outputs public parameters ${\sf pp}$, including the HE public/evaluation keys, and the $\Pi_{\sf HE}$ secret key $\sk$.
            \item ${\sf Query}(j,{\sf pp})\to{\sf q}$: Given public parameters and $\sf C$'s private index $j\in[p]$, outputs a query ${\sf q}$.
            \item ${\sf Emb}(M,{\sf q},{\sf pp})\to (\ct_k)_{k\in [r]}\in \cc C^r$: Given query, public parameters, and $\sf S$'s private table $M\in \mathcal T$, outputs an $r$-tuple of server-held $\Pi_{\sf HE}$ ciphertexts encrypted under $\sf C$'s key.
        \end{itemize}

\end{definition}
\begin{definition}[Correctness]
    $\Pi_{\sf PEL}(\Pi_{\sf HE},{\sf VecEcd},\mathcal T,E, u)$ is said to be \emph{correct} if, for any $j\in [p]$ and $M\in \mathcal T$, the following probability over the randomness of ${\sf Setup}$, ${\sf Query}$, and ${\sf Emb}$ is overwhelming in $\lambda$:
        \[
        \Pr\left[
        E\Big( M_j,\ {\sf VecDcd}\big( ({\sf Dec}(\sk,\ct_k))_{k\in[r]} \big)\Big) < 1/u
        \right],
        \]
        where $(\sk,{\sf pp})\gets{\sf Setup}(1^\lambda)$, ${\sf q}\gets{\sf Query}(j,{\sf pp})$, and $({\ct_k})_{k\in[r]}\gets{\sf Emb}(M,{\sf q},{\sf pp})$.
        The decryption in this condition is used only to define the plaintext represented by the server-held output ciphertexts.

\end{definition}
\begin{definition}[Security]
PEL satisfies \emph{client index privacy} if for any two distinct indices $j_1,j_2\in[p]$ and probabilistic polynomial time algorithm $\cc A$, the following is negligible in $\lambda$ for $(\sk,{\sf pp})\gets {\sf Setup}(1^\lambda)$ and ${\sf q}_i\leftarrow{\sf Query}(j_i,{\sf pp})$: 
   \[
        \Big|\Pr\left[1\gets\cc A({\sf q}_1, {\sf pp})\right]-\Pr\left[1\gets\cc A({\sf q}_2, {\sf pp})\right]\Big|.
        \]

\end{definition}

\paragraph{Threat Model}
We assume semi-honest adversaries. PEL targets \emph{client index privacy}: $\sf S$ should learn nothing about $j\in[p]$ from the public parameters and the protocol transcript.
In our one-message formulation, the client-to-server transcript consists only of ${\sf q}$; after receiving ${\sf q}$, ${\sf Emb}$ is local server-side post-processing using $(M,{\sf q},{\sf pp})$.
Thus indistinguishability of the query implies indistinguishability of the full server-side transcript by post-processing.
Since PEL does not return the resulting ciphertext to the client, the \emph{server privacy} of $M\in \mathcal T$ is automatically satisfied. Rather, server privacy should be interpreted in the context of the entire FHE pipeline that uses PEL as a component, rather than as a standalone property of PEL.
Related discussions appear in~\cite{OntheSecurityandPrivacyofCKKSbased,cryptoeprint:2025/395}, and we leave this as an orthogonal direction. 

\paragraph{Performance Metrics.}
We evaluate a PEL protocol by
\begin{enumerate}
  \renewcommand{\labelenumi}{(\roman{enumi})}

    \setlength{\topsep}{0pt}
  \setlength{\partopsep}{0pt}
  \setlength{\itemsep}{0pt}
  \setlength{\parsep}{0pt}
  \setlength{\parskip}{0pt}
  \setlength{\leftmargin}{1.2em} 
  \setlength{\labelsep}{0.4em}
  
  \item the client-to-server communication cost, measured by the query size $|{\sf q}|$ or the communication overhead $|{\sf q}|/\log p$,
  \item the server-side cost of ${\sf Emb}$ (computation cost), and
  \item the achieved accuracy captured by the error bound $u$.
\end{enumerate}

\section{IVE: Independent Vector Evaluation}\label{IVE}

In this section, we present the Independent Vector Evaluation (IVE) algorithm, a core method for our efficient, batched-PEL solution under the CKKS scheme.
We use a tilde to denote an approximate value (e.g., $\tilde{x}$), and denote the perturbation by $\epsilon_x := \tilde{x}-x$.
For an $m$-dimensional error vector $\epsilon$, we use the root-mean-square (RMS) norm $\|\epsilon\|_{\sf RMS}=\|\epsilon\|_2/\sqrt{m}$.

We start from the approach of~\citet{KPLC24}, where the client encrypts the index as a single-slot value, from which the server homomorphically constructs a one-hot vector ${\bf e}_j\in \bb R^p$ and computes $M_j= M {\bf e}_j$. 
Similarly, we take as our baseline the standard approach in which the client encodes the index and sends an encryption of the encoded value to the server for communication efficiency.

However, constructing a one-hot vector from an encrypted index is not arithmetic-friendly.
In particular, the approach of Kim et al.\ uses numerical methods: It relies on $O(p\log p)$ multiplicative operations to drive each component of the vector to converge to either $0$ or $1$, consuming substantial multiplicative depth and potentially incurring poor latency under CKKS.
One of the key principles in FHE-based applications is to reformulate computations to be as arithmetic-friendly as possible, so as to minimize multiplicative depth and improve latency.
From this perspective, directly synthesizing a one-hot vector from the encrypted index is not an ideal strategy. 

To overcome this limitation, we propose the Independent Vector Evaluation (IVE) method.
In IVE, the one-hot vectors ${\bf e}_j$ are replaced with linearly independent vectors ${\bf v}_j$ that serve as inputs to a matrix multiplication layer.
Let $L \in \mathbb{R}^{p\times p}$ denote the basis-change matrix satisfying $L {\bf v}_j = {\bf e}_j$, so that
\(
M_j = M {\bf e}_j = (M L) {\bf v}_j,
\)
where $ML$ can be precomputed offline.

\subsection{Choice of Independent Vectors}
Regarding the choice of $\{{\bf v}_j\}_{j \in [p]}$, two aspects are particularly important in this method: \emph{Computational Efficiency} and \emph{Numerical Stability}.

First, for the computational efficiency, since it suffices that $\{{\bf v}_j\}_{j \in [p]}$ are linearly independent, any efficient arithmetic circuit satisfying this condition may be chosen to maximize efficiency.

Next, since $M$ and $L$ are fixed and can be precomputed in plaintext with high precision, the dominant source of error arises from the evaluation of ${\bf v}_j$ from the encrypted index $j$ and multiplication between $ML$ and ${\bf v}_j$.
After the server obtains a ciphertext that decrypts to $\tilde {\bf v}_j = {\bf v}_j + {\epsilon_{{\bf v}_j}}$ for some error vector $\epsilon_{{\bf v}_j}$ after {IVE}, it homomorphically multiplies by $ML$ to produce an encryption of $\tilde M_j$. The corresponding error $\|\epsilon_{M_j}\|_{2}$ is bounded as follows:
\begin{equation}\label{align:Errorbound}
    \|\epsilon_{M_j}\|_{2}
    =\Vert \tilde M_j - M_j \Vert_{{2}}
    \leq  \Vert  \tilde M_j - ML\tilde {\bf v}_j\Vert_{2} +\Vert ML\epsilon_{{\bf v}_j}\Vert_{2}.
\end{equation}
The first term $\Vert \tilde M_j - ML\tilde {\bf v}_j  \Vert_{{2}}$ is determined by the stability of the matrix multiplication algorithm, and our primary concern lies in the second term, $\|ML\epsilon_{{\bf v}_j}\|_{2}$. It depends primarily on the error $\epsilon_{\bf v}$ incurred during the generation of ${\bf v}_j$ and on the conditioning of $M$ and $L$. 
In particular, if the basis-change matrix $L$ is ill-conditioned, the error $\epsilon_{{\bf v}_j}$ can be substantially amplified. 

\paragraph{Discrete Cosine Transform Matrix.} 
To meet the above two requirements, we propose to use the Discrete Cosine Transform (DCT) matrix~\cite{DCT}, an orthogonal matrix that has been widely adopted in image and signal processing due to its \emph{efficient implementation} and \emph{stable numerical properties}.
Our choice is further motivated by its compatibility with CKKS: since CKKS supports complex arithmetic, DCT entries can be generated from complex exponentials.
This also aligns naturally with CKKS functional bootstrapping, which efficiently produces such exponentials from encrypted exponents.

In this work, we assume $p$ is even and define a (row-permuted) DCT matrix $D=(D_{k,j}) \in \mathbb{R}^{p \times p}$ as follows.
\begin{equation*}
D_{k,j} =
\begin{cases}
    \sqrt{\tfrac{2}{p}}\,
    \cos\!\left( \tfrac{(-1)^j(2j+1)\pi}{2p}\,(k+1) \right)
    & 0 \leq k < \tfrac{p}{2} \\[1.2ex]
    \sqrt{\tfrac{2}{p}}\,
    \sin\!\left( \tfrac{(-1)^j(2j+1)\pi}{2p}\,\bigl(k-\tfrac{p}{2}+1\bigr) \right)
    & \tfrac{p}{2} \leq k < p-1 \\[1.2ex]
    \tfrac{1}{\sqrt{p}}
    & k = p-1
\end{cases}
\end{equation*}
We defer the proof of orthogonality of $D$ to Appendix~\ref{Appendix:DeferredProofs}.

Now let $L = D^{-1}$. It follows that $\|L\|_2=1$ due to orthogonality, and ${\bf v}_j$ corresponds to the $j$-th column of $D$. 

\subsection{Vector Evaluating Algorithm}\label{VEA}
We now describe how to compute $\mathbf{v}_j$ efficiently in CKKS.
\paragraph{Evaluating $\mathbf v_j$.}
Let $\alpha_j := \exp\!\left(J_j\cdot2\pi\sqrt{-1} \right)$ for $J_j:=\tfrac{(-1)^j(2j+1)}{4p}$.
Assume that the server is given $\alpha_j$. Then, the index-wise linearly independent vectors ${\bf v}_j$'s can be efficiently evaluated as follows: It holds that ${\bf v}_j={\bf v}(\alpha_j)$ where 
\begin{equation}\label{eq:valpha}
\textstyle
{\bf v}(\alpha) := { \sqrt{\tfrac{2}{p}} 
\Big(Re(\alpha, \dots, \alpha^{\tfrac{p}{2}}), 
Im(\alpha, \dots, \alpha^{\tfrac{p}{2}-1}), 
\tfrac{1}{\sqrt{2}}\Big)^\top}.
\end{equation}
We defer the proof for ${\bf v}_j={\bf v}(\alpha_j)$ to Appendix~\ref{Appendix:DeferredProofs}.

Algorithm~\ref{alg:IVE} describes the procedure for computing encryptions of a constant multiple of the $({\bf v}_{j_i})_{i\in[N/2]}$ from the given encrypted $(\alpha_{j_i})_{i\in[N/2]}$.
The algorithm outputs $\sqrt{2p}\cdot{\bf v}_{j_i}$; this constant factor is absorbed into the subsequent matrix multiplication phase to reduce depth consumption.
By packing $\alpha_{j_i}$ across slots of $\ct$, the algorithm exploits batch computation to evaluate all $\mathbf v_{j_i}$'s in parallel.
In the first stage, the algorithm performs power generation in a multiplication-tree manner to produce ciphertexts that encrypt successive powers of $\alpha_{j_i}$'s with logarithmic multiplicative depth.
For the stated tree schedule and cost table, we assume $p/2$ is a power of two; otherwise, the vector family can be padded to the next such size.
In the second stage, it separates the resulting values into their real and imaginary parts by conjugation, outputting ciphertexts that encrypt $(2Re(\alpha_{j_i}^{k+1}))_i$ and $(2Im(\alpha_{j_i}^{k+1}))_i$ for the corresponding zero-based loop index $k$.
We note that, in line~\ref{line:n} in Algorithm~\ref{alg:IVE}, multiplying an integer constant $(\sqrt{-1})^{-1}\in \bb Z[\sqrt{-1}]$ does not consume level.

\paragraph{Matrix Multiplication.} 
By executing the above procedure in CKKS, the scaled vectors $\sqrt{2p}\cdot\mathbf{v}_{j_i}$ for $i\in [N/2]$ are computed in a batched manner, yielding an encryption of a matrix, denoted by ${\sf CT}$.
Once ${\sf CT}=(\ct_k)_{k\in[p]}$ encrypting $\sqrt{2p}\cdot{\bf v}(\alpha_{j_i})$ is available, the server performs a plaintext--ciphertext matrix multiplication ($\sf PCMM$) with the preprocessed plaintext matrix $ML/\sqrt{2p}$.
This yields encryptions of the selected embeddings $(M_{j_i})_{i\in[N/2]}$ in a batched manner.
Table~\ref{tab:ive-cost} summarizes the CKKS resources consumed in this process.

\algnewcommand\algorithmicinput{\textbf{Input:}}
\algnewcommand\algorithmicoutput{\textbf{Output:}}
\algnewcommand\Input{\item[\algorithmicinput]}
\algnewcommand\Output{\item[\algorithmicoutput]}


\begin{algorithm}[ht]
\caption{\textsf{IVE} algorithm}
\label{alg:IVE}
\begin{algorithmic}[1]
\Input $\ct$: $\dec(\ct)[i]\approx\alpha_{j_i}$.
\Output $(\ct_{k})_{k\in[p]}$: $\dec(\ct_k)[i] \approx {\sqrt{2p}}\cdot\mathbf v_{j_i}[k]$.
\State $n\gets p/2$
\State ${\ct'}_0\gets \ct$

\For{$d=0$ \textbf{to} $\log_2 n-1$}
    \State $i\gets 2^d$
    \State $t\gets \min(i,\,n-i)$
    \For{$j=1$ \textbf{to} $t$}
        \State ${\ct'}_{i+j-1}\gets {\ct'}_{i-1}\cdot {\ct'}_{j-1}$
    \EndFor
\EndFor

\For{$k=0$ \textbf{to} $n-1$}
    \State $\overline{\ct} \gets \overline{\ct'_k}$
    \State $\ct_k \gets {\ct}'_k + \overline{\ct}$
    \If{$k < n-1$}
        \State $\ct_{n+k} \gets \big({\ct}'_k-\overline{\ct}\big)/\sqrt{-1}$\label{line:n}
    \Else
        \State $\ct_{p-1}\gets {\sf Enc}(\sqrt{2})$
    \EndIf
\EndFor

\State \Return $(\ct_k)_{k\in[p]}$
\end{algorithmic}
\end{algorithm}

\begin{table}[ht]
\centering
\caption{Resource Consumption of Homomorphic Evaluation on the Server-side for {\sf IVE}, plus one {\sf PCMM} call.}
\label{tab:ive-cost}
{
\begin{tabular}{cccccc}
\toprule
\textbf{\#Add} & \textbf{\#ptMult} & \textbf{\#Mult} & \textbf{\#Conj} & \textbf{\#PCMM} & \textbf{Depth} \\
\midrule
$p-1$ & $\dfrac{p}{2}-1$ & $\dfrac{p}{2}-1$ & $\dfrac p2$ & $1$ & $\log p$ \\
\bottomrule
\end{tabular}
}
\end{table}

\section{IVE-PEL}\label{sec:ive-pel}
In the previous section, we showed how the server can homomorphically derive the target embedding vector $M_j$
once it is given an encryption of $\alpha_j$.
This section completes the construction by describing how the client can provide such an input ciphertext
with compact amortized communication while achieving the \emph{target numerical precision} in the output embedding vector.
For clarity, we describe the transmission for a batch size of $N/2$, where $N$ is the ciphertext ring degree, and the same approach applies to other batch sizes.

\subsection{Query Transmission}
\paragraph{Standard Transmission Layer}
We use standard RLWE ciphertext-transmission techniques~\cite{Hermes,10.1007/978-3-030-78372-3_18} as a communication layer.
The client forms the values $\alpha_{j_i}$, as defined in Section~\ref{VEA}, encrypts them at the lowest modulus $Q_0$ over $\mathcal{R}_{N/2,Q_0}^{2}$, and sends a seed for the $a$ component together with the $b$ component.
The server regenerates $a$ and packs the ciphertext to $\mathcal{R}_{N,Q_0}^{2}$ via ring packing, obtaining the input ciphertext for IVE.
We defer the details of ring packing to Appendix~\ref{app:A.ringpacking}.

\paragraph{IVE-PEL}
Putting the above components together, we obtain the \textsf{IVE-PEL} protocol with batch size $N/2$.
We summarize the construction in Figure~\ref{fig:IVE-PEL Protocol}. In particular, the server-side computation (${\sf Emb}$) is depicted in Figure~\ref{FigIVE}.

\begin{figure}[!ht]
\centering
\fbox{%
\begin{minipage}{0.97\linewidth}
  \centering
  \textbf{Protocol \textsf{IVE-PEL}}\par\vspace{-0.2em}
  \begin{flushleft}
  {\bf Input}: $\sf C$ holds $(j_i)_{i\in[N/2]}\in[p]^{N/2}$, $\sf S$ holds $M\in \bb R^{d\times p}$.\\
  {\bf Output}: $\sf S$ gets CKKS encryptions of ${M_{j_i}}\in \bb R^d$ for $i\in[N/2]$.
  \end{flushleft}
\vspace{-1em}
  \rule{\linewidth}{0.4pt}
  \vspace{-1em}
  \begin{itemize}
    \item ${\sf Setup}(1^\lambda)$: Share $\sf pp$ including CKKS public/evaluation parameters and ring-packing keys. $\sf C$ gets CKKS secret $\sk$.
    \item ${\sf Query}((j_i)_i,\sf pp)$: $\sf C$ samples a 128-bit seed $\rho$, uses it to generate $a$, and encrypts $\left(\alpha_{j_i}\right)_{i\in[N/2]}$ at $0$ level to get $(a,b)\in \cc R_{N/2,Q_0}^2$.
    Outputs query ${\sf q}=(\rho,b)$ and sends to $\sf S$.
    \item ${{\sf Emb}(M,{\sf q},{\sf pp})}$: $\sf{S}$ reconstructs $(a,b)$ and packs into ring degree $N$ to get $\ct\in \cc R_{N,Q_0}^2$.
    Then computes ${(\ct_{{\sf out},k})}_{k\in[d]}\gets({\sf PCMM}_{ML/\sqrt{2p}}\circ{\sf IVE})(\ct)$
  \end{itemize}
\end{minipage}}
\caption{{\sf IVE-PEL} protocol.}
\Description{Protocol box for IVE-PEL.}
\label{fig:IVE-PEL Protocol}
\end{figure}

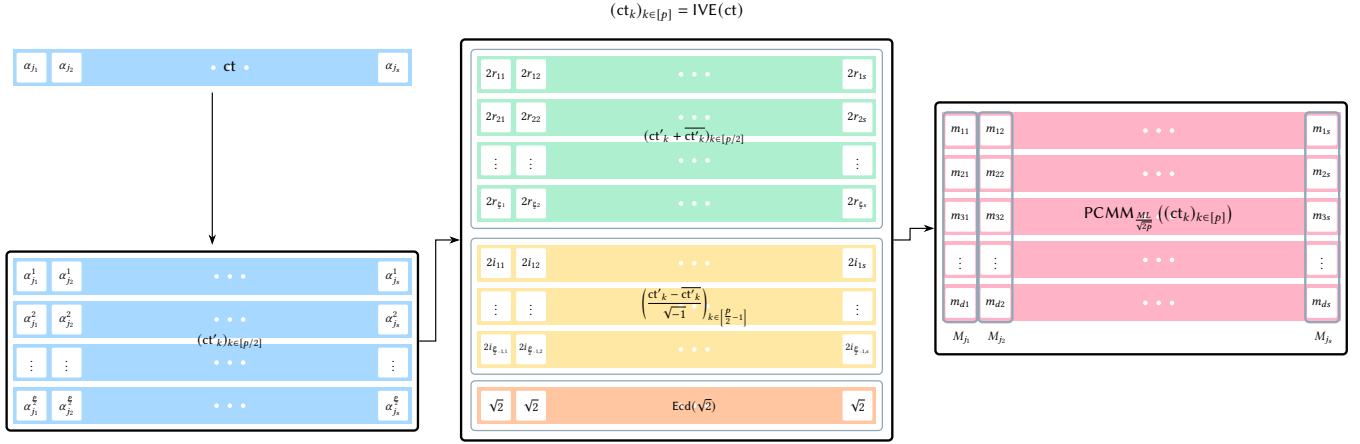
\begin{figure*}[t]
\centering
\resizebox{\linewidth}{!}{%
\begin{tikzpicture}[x=1mm,y=1mm,font=\small]

\def\rowW{120}       
\def\rowH{11}       
\def\rowSep{2}      
\def\slotW{9}       
\def\slotH{9}       
\def\slotGap{1.6}   
\def\dotGap{4}      
\def\framePad{2}    

\definecolor{blkBlue}{HTML}{A7D8FF}
\definecolor{blkBlueB}{HTML}{3F8FBF}

\definecolor{blkGreen}{HTML}{AEEFD0}
\definecolor{blkGreenB}{HTML}{3D9B6B}

\definecolor{blkYellow}{HTML}{FFE8A3}
\definecolor{blkYellowB}{HTML}{C9A63A}

\definecolor{blkOrange}{HTML}{FFC6A1}
\definecolor{blkOrangeB}{HTML}{C7783E}

\definecolor{blkRed}{HTML}{FFB3C7}
\definecolor{blkRedB}{HTML}{C75A7A}

\definecolor{frameMain}{HTML}{A46C73}   

\definecolor{frameSub}{HTML}{95A7B6}   
\definecolor{frameSubLight}{HTML}{9AA5B1}

\tikzset{
  slot/.style={fill=white, minimum width=\slotW mm, minimum height=\slotH mm, inner sep=0pt, line width=0.35pt, rounded corners=0.5mm, font=\Large},
  dot/.style ={circle, fill=white, inner sep=0pt, minimum size=1.4mm},
  frame/.style={draw=frameSub, line width=1.1pt, rounded corners=1.2mm, inner sep=\framePad mm, fill=none},
  framea/.style={draw=black, line width=2pt, rounded corners=1.2mm, inner sep=\framePad mm, fill=none},
  arr/.style  ={-{Stealth[length=3.2mm,width=2.2mm]}, line width=1pt},
}

\newcommand{\CipherRow}[7][\Large]{%
  \pgfmathsetmacro{\midY}{0.5*\rowH}%
  \pgfmathsetmacro{\midX}{0.5*\rowW}%
  \begin{scope}[shift={(#2,#3)}]
    \path[fill=#4, draw=none] (0,0) rectangle (\rowW,\rowH);

    \foreach \i/\txt in {0/{#5},1/{#6}}{%
      \pgfmathsetmacro{\sx}{1+0.5*\slotW + \i*(\slotW+\slotGap)}%
      \node[slot, font=#1] at (\sx,\midY) {\txt};%
    }%

    \pgfmathsetmacro{\sxR}{\rowW - (1+0.5*\slotW)}%
    \node[slot, font=#1] at (\sxR,\midY) {#7};%

    \foreach \dx in {-\dotGap,0,\dotGap}{%
      \pgfmathsetmacro{\px}{0.5*(\slotW+\slotGap)+\midX + \dx}%
      \node[dot] at (\px,\midY) {};%
    }%
  \end{scope}%
}

\newcommand{\CipherRowa}[5]{%
  \pgfmathsetmacro{\midY}{0.5*\rowH}%
  \pgfmathsetmacro{\midX}{0.5*\rowW}%
  \begin{scope}[shift={(#1,#2)}]
    \path[fill=#3, draw=none] (0,0) rectangle (\rowW,\rowH);

    \foreach \i/\txt in {0/{#4}}{%
      \pgfmathsetmacro{\sx}{5+0.5*\slotW + \i*(\slotW+\slotGap)}%
      \node[slot] at (\sx,\midY) {\txt};%
    }%

    \pgfmathsetmacro{\sxR}{\rowW - (5+0.5*\slotW)}%
    \node[slot] at (\sxR,\midY) {#5};%

    \foreach \dx in {-\dotGap,0,\dotGap}{%
      \pgfmathsetmacro{\px}{\midX + \dx}%
      \node[dot] at (\px,\midY) {};%
    }%
  \end{scope}%
}

\def\leftX{0}
\def\leftY{-12}

\def\topX{0}
\def\topY{100}
\pgfmathsetmacro{\topMidY}{\topY + 0.5*\rowH}
\pgfmathsetmacro{\topMidX}{\topX + 0.5*\rowW}

\path[fill=blkBlue, draw=none] (\topX,\topY-10) rectangle (\topX+\rowW,\topY+\rowH-10);

\foreach \i/\txt in {0/{$\alpha_{j_1}$},1/{$\alpha_{j_2}$}}{
  \pgfmathsetmacro{\sx}{1+\topX + 0.5*\slotW + \i*(\slotW+\slotGap)}
  \node[slot,font=\Large] at (\sx,\topMidY-10) {\txt};
}

\foreach \i/\txt in {0/{$\alpha_{{j_s}}$}}{
  \pgfmathsetmacro{\sx}{\topX + \rowW - (1+0.5*\slotW + \i*(\slotW+\slotGap))}
  \node[slot] at (\sx,\topMidY-10) {\txt};
}

\foreach \dx in {-5,0,5}{
  \pgfmathsetmacro{\px}{0.5*(\slotW+\slotGap)+\topMidX + \dx}
  \node[dot] at (\px,\topMidY-10) {};
}

\node[font=\Large, opacity=0] (ctsum) at (\leftX+0.5*\rowW + 0.5*\slotW + 0.5*\slotGap,\topY-4){$\ct$};



\node[font=\huge] at (\leftX+0.5*\rowW + 0.5*\slotW + 0.5*\slotGap,\topY-4) {$\ct$};


\pgfmathsetmacro{\CTlabelY}{\leftY + 0.5*(4*\rowH + 3*\rowSep)}


  \pgfmathsetmacro{\yy}{\leftY + 3*(\rowH+\rowSep)}
  \CipherRow{\leftX}{\yy}{blkBlue}{$\alpha_{j_1}^1$}{$\alpha_{j_2}^1$}{$\alpha_{j_s}^1$}
  \pgfmathsetmacro{\yy}{\leftY + 2*(\rowH+\rowSep)}
  \CipherRow{\leftX}{\yy}{blkBlue}{$\alpha_{j_1}^2$}{$\alpha_{j_2}^2$}{$\alpha_{j_s}^2$}
  \pgfmathsetmacro{\yy}{\leftY + 1*(\rowH+\rowSep)}
  \CipherRow{\leftX}{\yy}{blkBlue}{$\vdots$}{$\vdots$}{$\vdots$}
  \pgfmathsetmacro{\yy}{\leftY + 0*(\rowH+\rowSep)}
  \CipherRow{\leftX}{\yy}{blkBlue}
  {$\displaystyle \alpha_{j_1}^{\scalebox{0.9}{$\scriptscriptstyle{\frac p2}$}}$}
  {$\displaystyle \alpha_{j_2}^{\scalebox{0.9}{$\scriptscriptstyle{\frac p2}$}}$}
  {$\displaystyle \alpha_{j_s}^{\scalebox{0.9}{$\scriptscriptstyle{\frac p2}$}}$}

\pgfmathsetmacro{\lblXb}{\leftX+0.5*\rowW + 0.5*\slotW + 0.5*\slotGap}
\pgfmathsetmacro{\lblYb}{\rowH+\rowSep}

\node[font=\LARGE, opacity=0] (ctsum) at (\lblXb,\lblYb)
  {$({\sf ct'})$};



\node[font=\LARGE
] at (\lblXb,\lblYb)
  {$({\sf ct'}_k)_{k\in[p/2]}$};

\pgfmathsetmacro{\CTtopY}{\leftY + 4*\rowH + 3*\rowSep}
\node[framea, fit={(\leftX,\leftY) (\leftX+\rowW,\CTtopY)}] (CTframe) {};

\pgfmathsetmacro{\arrTopY}{\topY - 12}
\pgfmathsetmacro{\arrBotY}{\CTtopY + 4}
\draw[arr] (\topMidX,\arrTopY) -- (\topMidX,\arrBotY);

\def\rightX{140}
\def\rightY{10}


  \pgfmathsetmacro{\yy}{\rightY + (3)*(\rowH+\rowSep) + 3*(\rowH+\rowSep)}
  \CipherRow{\rightX}{\yy}{blkGreen}{$2r_{11}$}{$2r_{12}$}{$2r_{1s}$}
  
    \pgfmathsetmacro{\yy}{\rightY + (2)*(\rowH+\rowSep) + 3*(\rowH+\rowSep)}
  \CipherRow{\rightX}{\yy}{blkGreen}{$2r_{21}$}{$2r_{22}$}{$2r_{2s}$}

  \pgfmathsetmacro{\yy}{\rightY + (1)*(\rowH+\rowSep) + 3*(\rowH+\rowSep)}
  \CipherRow{\rightX}{\yy}{blkGreen}{$\vdots$}{$\vdots$}{$\vdots$}

  \pgfmathsetmacro{\yy}{\rightY + (0)*(\rowH+\rowSep) + 3*(\rowH+\rowSep)}
  \CipherRow{\rightX}{\yy}{blkGreen}
  {$\displaystyle 2r_{\scalebox{0.7}{$\scriptscriptstyle{\frac p2}$}\scalebox{0.9}{$\scriptscriptstyle{1}$}}$}
  {$\displaystyle 2r_{\scalebox{0.7}{$\scriptscriptstyle{\frac p2}$}\scalebox{0.9}{$\scriptscriptstyle{2}$}}$}
  {$\displaystyle 2r_{\scalebox{0.7}{$\scriptscriptstyle{\frac p2}$}\scalebox{0.9}{$\scriptscriptstyle{s}$}}$}

\pgfmathsetmacro{\Gbot}{\rightY + 3*(\rowH+\rowSep)}
\pgfmathsetmacro{\Gtop}{\Gbot + 4*\rowH + 3*\rowSep}
\node[frame, fit={(\rightX,\Gbot) (\rightX+\rowW,\Gtop)}] (Gframe) {};

\pgfmathsetmacro{\lblX}{\rightX + 0.5*\rowW + 0.5*\slotW + 0.5*\slotGap}
\pgfmathsetmacro{\lblY}{\Gbot + 26}

\node[font=\LARGE, opacity=0] (ctsum) at (\lblX,\lblY)
  {$\sf CT'+\overline{CT'}$};



\node[font=\LARGE] at (\lblX,\lblY)
  {$({\sf ct'}_k+\overline{{\sf ct'}_k})_{k\in[p/2]}$};


  \pgfmathsetmacro{\yy}{-5 + \rightY + (2)*(\rowH+\rowSep)}
  \CipherRow{\rightX}{\yy}{blkYellow}{$2i_{11}$}{$2i_{12}$}{$2i_{1s}$}
  
  \pgfmathsetmacro{\yy}{-5 + \rightY + (1)*(\rowH+\rowSep)}
  \CipherRow{\rightX}{\yy}{blkYellow}{$\vdots$}{$\vdots$}{$\vdots$}
  
  \pgfmathsetmacro{\yy}{-5 + \rightY + (0)*(\rowH+\rowSep)}
  \CipherRow[\normalsize]{\rightX}{\yy}{blkYellow}
  {$\displaystyle 2i_{\scalebox{0.9}{$\scriptscriptstyle{\frac p2}$}\scalebox{0.6}{$\scriptscriptstyle{-}$}\scalebox{0.9}{$\scriptscriptstyle{1,1}$}}$}
  {$\displaystyle 2i_{\scalebox{0.9}{$\scriptscriptstyle{\frac p2}$}\scalebox{0.6}{$\scriptscriptstyle{-}$}\scalebox{0.9}{$\scriptscriptstyle{1,2}$}}$}
  {$\displaystyle 2i_{\scalebox{0.9}{$\scriptscriptstyle{\frac p2}$}\scalebox{0.6}{$\scriptscriptstyle{-}$}\scalebox{0.9}{$\scriptscriptstyle{1,s}$}}$}

\pgfmathsetmacro{\Ytop}{-5+\rightY + 3*\rowH + 2*\rowSep}
\node[frame, fit={(\rightX,-5+\rightY) (\rightX+\rowW,\Ytop)}] (Yframe) {};

\pgfmathsetmacro{\lblXy}{\rightX + 0.5*\rowW + 0.5*\slotW + 0.5*\slotGap}
\pgfmathsetmacro{\lblYy}{\Ytop -18.5}

\node[font=\Large, opacity=0] (ctsum) at (\lblXy,\lblYy)
  {$\dfrac{\sf CT'-\overline{CT'}}{\sqrt{-1}}[0;\tfrac p2-2]$};



\node[font=\Large] at (\lblXy,\lblYy)
  {$\left(\dfrac{{\sf ct'}_k-\overline{{\sf ct'}_k}}{\sqrt{-1}}\right)_{k\in\left[\tfrac p2-1\right]}$};

\def\botY{-12}
\pgfmathsetmacro{\Otop}{\botY + \rowH}
\CipherRow{\rightX}{\botY}{blkOrange}{$\sqrt 2$}{$\sqrt 2$}{$\sqrt 2$}
\node[frame, fit={(\rightX,\botY) (\rightX+\rowW,\Otop)}] (Oframe) {};

\node[framea, fit={(\rightX-3,-3+\botY) (\rightX+\rowW+3,\Gtop+3)}] (wholeframe) {};
\node[anchor=north,font=\huge] at (\rightX+0.5*\rowW,\Gtop+17) {${(\ct_k)_{k\in[p]}={\sf IVE}(\ct)}$};
\draw[arr] (CTframe.east) -- ++(5,0) |- (wholeframe.west);

\pgfmathsetmacro{\lblXy}{\rightX + 0.5*\rowW + 0.5*\slotW + 0.5*\slotGap}
\pgfmathsetmacro{\lblYr}{\Ytop -48.5}

\node[font=\Large, opacity=0] (ctsum) at (\lblXy,\lblYr)
  {${\sf Ecd}(\sqrt 2)$};
  


\node[font=\Large] at (\lblXy,\lblYr)
  {${\sf Ecd}(\sqrt 2)$};

\def\newrightX{280}

\node[framea, fit={(\newrightX,\botY+23) (\newrightX+\rowW,\Gtop-16)}] (newframe) {};
\draw[arr] (wholeframe.east) -- ++(5,0) |- (newframe.west);

\pgfmathsetmacro{\yy}{-4 + \rightY + (5)*(\rowH+\rowSep)}
\CipherRow{\newrightX}{\yy}{blkRed}{$m_{11}$}{$m_{12}$}{$m_{1s}$}
\pgfmathsetmacro{\yy}{-4 + \rightY + (4)*(\rowH+\rowSep)}
\CipherRow{\newrightX}{\yy}{blkRed}{$m_{21}$}{$m_{22}$}{$m_{2s}$}
\pgfmathsetmacro{\yy}{-4 + \rightY + (3)*(\rowH+\rowSep)}
\CipherRow{\newrightX}{\yy}{blkRed}{$m_{31}$}{$m_{32}$}{$m_{3s}$}
\pgfmathsetmacro{\yy}{-4 + \rightY + (2)*(\rowH+\rowSep)}
\CipherRow{\newrightX}{\yy}{blkRed}{$\vdots$}{$\vdots$}{$\vdots$}
\pgfmathsetmacro{\yy}{-4 + \rightY + (1)*(\rowH+\rowSep)}
\CipherRow{\newrightX}{\yy}{blkRed}{$m_{d1}$}{$m_{d2}$}{$m_{ds}$}

\node[frame, fit={(\newrightX+2.4,\botY+32.7)(\newrightX+8.6,\Gtop-19)},draw=frameSub,line width = 2pt]{};
\node[font=\Large] at (\newrightX+6,\botY+26)
  {${M_{j_1}}$};

  \node[frame, fit={(\newrightX+\slotW+\slotGap+2.4,\botY+32.7)(\newrightX+\slotW+\slotGap+8.6,\Gtop-19)},draw=frameSub,line width = 2pt]{};
\node[font=\Large] at (\newrightX+\slotW+\slotGap+6,\botY+26)
  {${M_{j_2}}$};

  \node[frame, fit={(\newrightX+\rowW-\slotW-\slotGap+2.1,\botY+32.7)(\newrightX+\rowW-\slotW-\slotGap+8.3,\Gtop-19)},draw=frameSub,line width = 2pt]{};
\node[font=\Large] at (\newrightX+\rowW-\slotW-\slotGap+5.7,\botY+26)
  {${M_{j_s}}$};

\pgfmathsetmacro{\lblXp}{\newrightX + 0.5*\rowW + 0.5*\slotW + 0.5*\slotGap}
\pgfmathsetmacro{\lblYp}{-4 + \rightY + (3.4)*(\rowH+\rowSep)}

\node[font=\Large, opacity=0] (ctsum) at (\lblXp,\lblYp)
  {${\sf Ecd}(\sqrt 2)$};



\node[font=\huge] at (\lblXp,\lblYp)
  {${{\sf PCMM}_{\frac{ML}{\sqrt{2p}}}\left((\ct_k)_{k\in[p]}\right)}$};

\end{tikzpicture}
}
\caption{IVE-PEL server-side procedure ($\sf Emb$).
Given a ciphertext $\ct$ encoding $(\alpha_{j_m})_{m\in[s]}$, the server obtains the embedding vectors $(M_{j_m})_{m\in[s]}$ by applying $\sf IVE$ and $\sf PCMM$.
Each row corresponds to a single CKKS ciphertext. Here, $s$ denotes the number of slots (i.e., $s=N/2$). For $1\le \ell\le p/2$ and $1\le m\le s$, $r_{\ell m}$ and $i_{\ell m}$ denote $\operatorname{Re}(\alpha_{j_m}^{\ell})$ and $\operatorname{Im}(\alpha_{j_m}^{\ell})$, respectively.}
\Description{IVE-PEL server-side procedure showing IVE and PCMM.}
\label{FigIVE}
\end{figure*}

\subsection{Numerical Precision and Parameter Selection}\label{SUBSEC:NUMERICAL}
The scaling factor $\Delta$ is chosen to meet the target precision and error budget.
We state a sufficient condition below and defer the derivation to Appendix~\ref{Appendix:MainThm}.

\paragraph{Error Model and Notation.}
Fix a correctness parameter $u$, an even index size $p$ satisfying the assumptions of Section~\ref{IVE} (or its padded size), an embedding-vector dimension $d$, and a CKKS ring degree $N$.
Let $\Delta$ be the CKKS scale factor.
Let $C_{\sf PCMM}$ bound the final PCMM RMS error by $C_{\sf PCMM}\Delta^{-1}$, and let $C_T$ satisfy $\|M\|_2\le C_T(\sqrt d+\sqrt p)$ for tables $M\in\mathcal T$.
Let $C_V$ be a local Lipschitz constant for the map $\alpha\mapsto{\bf v}(\alpha)$.
Let $C_{\rm coeff}=19.2$ denote the coefficient-domain 6-sigma bound for Gaussian error of standard deviation $\sigma=3.2$, and set
$C_0:=C_{\rm coeff}\sqrt N$ for the corresponding canonical-embedding bound, following the CKKS noise analysis~\cite{CKKS}.
Assume that after ring packing, the value entering IVE is
$\tilde\alpha_i=\alpha_{j_i}+\epsilon_{\alpha,i}$ with
\[
\left|\epsilon_{\alpha,i}\right|
\le
\frac{C_0}{\Delta}.
\]
The expanded RMS error budget is
\begin{equation}\label{eq:sufficient-condition}
\frac{C_{\sf PCMM}}{\Delta}
+
\frac{C_T(\sqrt d+\sqrt p)}{\sqrt d}
\cdot
C_V\frac{p+2}{\sqrt{12}}
\cdot
\frac{C_0}{\Delta}
\leq
\frac{1}{u},
\end{equation}

\begin{theorem}\label{thm:admissible-params}
Under the error model above, suppose~\eqref{eq:sufficient-condition} holds.
Then, with overwhelming probability, \textsf{IVE-PEL} returns encryptions of $M_{j_i}$ for every packed index $j_i$, with RMS error at most $1/u$.
\end{theorem}

\paragraph{Interpreting the condition.}
Define
\[
S(\Delta)
:=
\left(\frac{\Delta}{u}-C_{\sf PCMM}\right)
\frac{\sqrt d}{C_V C_T(\sqrt d+\sqrt p)}
\frac{\sqrt{12}}{p+2}.
\]
For $\Delta/u>C_{\sf PCMM}$, the condition becomes
\begin{equation}\label{eq:s-delta-condition}
C_0\le S(\Delta).
\end{equation}
Equivalently, it suffices to choose
\[
\Delta
\ge
uC_{\sf PCMM}
+
u\cdot
\frac{C_V C_T(\sqrt d+\sqrt p)}{\sqrt d}\cdot
\frac{p+2}{\sqrt{12}}\cdot
C_0.
\]
If $C_{\sf PCMM}=o(\Delta/u)$, then this scale satisfies
\begin{equation}\label{eq:gap-scaling}
\log\Delta
\;=\;
\log\;\!\Big( u\cdot p\cdot C_0\cdot\big(1+\sqrt{p/d}\ \big)\Big) + O(1),
\end{equation}
where $O(1)$ hides absolute constants and table-dependent factors.

\paragraph{Query-size Accounting}
We generate the $a$ component from a 128-bit seed and transmit the $b$ component at modulus $Q_0$.
Sending one encrypted scalar query slot for each of the $N/2$ indices costs
\(
128 + ({N}/{2})\log Q_0
\) bits.
Thus the amortized multiplicative overhead relative to the plaintext index representation is
\[
\frac{128+({N}/{2})\log Q_0}{({N}/{2})\log p}
=
\frac{\log Q_0}{\log p}
+\frac{256}{N\log p}.
\]
With $Q_0=\Delta$ and $\Delta$ chosen as above, this amounts to
\begin{equation}\label{eq:overhead} 
1+\frac{\log\Bigl( u\cdot C_0\cdot (1+\sqrt{p/d}) \Bigr)}{\log p}
+O\!\left(\frac{1}{\log p}\right)
+\frac{256}{N\log p},
\end{equation}
which is bounded with respect to $p$ for fixed target precision.

\section{Implementation}\label{Section: Implementation}
In this section, we evaluate the practical implementation performance of private embedding lookup. The experiments are divided into two parts. In the first experiment, presented in Section~\ref{Section: Implementation1}, we compare our method with the private embedding lookup method of Kim et al.~\cite{KPLC24}, which was the most efficient prior approach at the embedding-layer level for lookup from an encrypted index. This experiment focuses on the cost difference in vector generation. In the second experiment, presented in Section~\ref{Section: Implementation2}, we integrate our lookup method into an end-to-end encrypted inference pipeline and evaluate its practicality at the application level.
For this purpose, we use a FastText-style text classification model on datasets representing practical text-classification workloads, where lightweight embedding-based classification is a natural fit and the embedding layer accounts for a significant part of the computation.
The results show that our method can efficiently process private token queries while preserving accurate classification performance.

Throughout the experiments, we consider a setting where the client sends many private token queries to the server at once. This setting naturally arises when a document contains many tokens, when multiple documents are processed together, or when several inference requests are batched. In such cases, CKKS SIMD packing allows the cost of homomorphic operations to be amortized across many queries. Therefore, in our implementation, we apply plaintext-ciphertext matrix multiplication (PCMM)~\cite{PCMM} to multiply the independent vectors generated from multiple queries by the embedding table simultaneously, thereby reducing the cost of the linear part. We separately report the case where the number of packed queries is smaller than the number of CKKS SIMD slots in Appendix~\ref{Appendix:SmallQueries}.

\paragraph{Implementation Environment}
Our CKKS implementations are built on the C++ HEaaN library~\cite{Github_HEaaN} and executed on an Intel Xeon Gold 6542Y with a single CPU thread.
For the linear transformation, we use PCMM algorithm~\cite{PCMM}, while plaintext matrix multiplications are computed with the single-threaded OpenBLAS library~\cite{OpenBLAS}.
Training experiments are performed on an NVIDIA GeForce RTX 4090 GPU.

\subsection{Evaluation of Private Embedding Lookup}\label{Section: Implementation1}
We first compare Kim et al.~\cite{KPLC24} and our IVE-PEL as private embedding lookup methods.
This experiment evaluates the server-side cost of performing private embedding lookup from encrypted token indices and shows the performance gain of replacing one-hot vector generation with IVE.

\subsubsection{Implementation Detail}

\paragraph{Embedding Datasets}
We use three widely used embedding tables as source pretrained embeddings:
\begin{itemize}
    \item \textbf{GloVe.42B.300d}~\cite{Glove}: $P=1{,}917{,}494$ words, $d=300$.
    \item \textbf{GloVe.6B.50d}~\cite{Glove}: $P=400{,}000$ words, $d=50$.
    \item \textbf{GPT-2}~\cite{GPT2}: $P=50{,}257$ tokens, $d=768$.
\end{itemize}
Here, $P$ denotes the number of distinct token entries in the source embedding table, and $d$ denotes the embedding dimension, i.e., the size of each embedding vector.
These source embedding tables are used to train compressed embedding
representations offline.

\paragraph{Embedding Table Compression}
As discussed in Section~\ref{subsec:intro-design-space}, table compression is
used as a front-end before private lookup is applied.
For each source embedding table
$M\in\mathbb{R}^{P\times d}$, we adopt the table-compression method
of~\cite{Shu} and follow the implementation in~\cite{Github} for training.
The method represents $M$ using $\ell$ compressed subtables
$M^{(1)},\ldots,M^{(\ell)}\in\mathbb{R}^{p\times d}$.
Each original token index $j\in[P]$ is represented by an index tuple
$(j_1,\ldots,j_\ell)\in[p]^\ell$, and its embedding is reconstructed as
\[
M_j \approx \sum_{i=1}^{\ell} M^{(i)}_{j_i}.
\]
Thus, the embedding table size is reduced from $P\times d$ to
$\ell p\times d$.
In all experiments in this subsection, we fix $\ell=4$ and vary the subtable
size $p$.

To assess whether the compressed embeddings preserve downstream utility, we
report IMDB~\cite{IMDB} sentiment classification accuracy.
For this evaluation, we train a lightweight LSTM classifier using the
compressed GloVe embeddings as pretrained embeddings.
Detailed compression training settings and IMDB classifier training details are
provided in Appendix~\ref{Appendix:Experiment}.

\paragraph{Experimental Setup}
For each compressed embedding table and each choice of subtable size $p$, we
compare the method of Kim et al.~\cite{KPLC24} and our IVE-PEL under the same
compressed tables, CKKS parameters, and PCMM algorithm. Because both methods use the identical scalar-index query, their client-to-server communication is the same, and we therefore focus the comparison on server-side cost.
We decompose the PEL pipeline into two parts: $\sf VecGen$ and $\sf PCMM$.
The $\sf VecGen$ stage generates the encrypted one-hot vectors in Kim et al.'s
method and the linearly independent vectors in ours.
The $\sf PCMM$ stage performs plaintext-ciphertext matrix multiplication between
the corresponding plaintext table and the generated encrypted vectors.
In our method, this plaintext table refers to the preprocessed matrix obtained
by multiplying the embedding matrix with the basis-change matrix $L$.\footnote{
This preprocessing is performed entirely by the server in plaintext and is a
one-time offline cost. In the same single-threaded CPU setting as our
experiments, it takes only $0.54$ seconds even for the setting $\log p=10$,
$d=300$, and $\ell=4$.}
Table~\ref{tab:compare} reports the amortized per-token time of $\sf VecGen$,
$\sf PCMM$, and the total PEL pipeline, together with the multiplicative depth
consumed by PEL, for both methods across different datasets and values
of $p$.

\paragraph{CKKS Parameter.}
To fairly compare the algorithmic cost of our method with that of Kim et al., we use the same CKKS parameter set for both methods. The parameters are chosen to support the method of Kim et al., which requires a large multiplicative depth. The concrete parameter set is given in Table~\ref{tab:ckks_param}. The resulting parameter set provides at least 128-bit security according to the Lattice Estimator~\cite{Lattice_Est}.

\subsubsection{Implementation Results}
Table~\ref{tab:compare} shows that our method reduces the cost of $\sf VecGen$
and thereby lowers the total cost of private embedding lookup.
Across all datasets and all values of $p$, the main performance difference between the two methods appears in the $\sf VecGen$ stage.
Recall that Kim et al.'s $\sf VecGen$ requires $O(p\log p)$ homomorphic operations for lookup size $p$, while ours requires only $O(p)$.
The results in Table~\ref{tab:compare} are consistent with this complexity difference, especially for larger values of $p$.


For example, on GloVe.6B.50d, the amortized $\sf VecGen$ time of Kim et al.'s
method increases from $11.7711\,$ms at $\log p=6$ to $424.7096\,$ms at $\log p=10$.
In contrast, the corresponding cost of IVE-PEL increases only from $0.1953\,$ms to $3.1543\,$ms.
This gives a $\sf VecGen$ speedup of up to $134.6\times$.
Consequently, for $\log p=10$ on GloVe.6B.50d, IVE-PEL reduces the total amortized PEL time from $427.0195\,$ms to $5.4461\,$ms, achieving a $78.4\times$ total speedup.
For the same $\log p=10$ setting, the total speedup is $36.4\times$ on GloVe.42B.300d.
For GPT-2, where the embedding dimension is larger ($d=768$), the $\sf PCMM$ stage accounts for a larger fraction of the total time, but IVE-PEL still achieves an $18.5\times$ total speedup.



A larger subtable size can better preserve the quality of the compressed embedding representation, as reflected by the IMDB accuracy results, but it also increases the cost of private lookup.
IVE-PEL reduces this cost increase by substantially lowering the $\sf VecGen$ cost.
In addition, our method consumes only about $1/3$ of the multiplicative depth required by Kim et al.'s method across the tested values of $p$.
Thus, IVE-PEL makes larger compressed subtables more practical for private embedding lookup.

\begin{table*}[t]
    \centering
    \caption{Comparison of IVE-PEL and Kim et al. on three datasets.
    For all experiments, each token is represented by four subtable indices,
    and $2^{16}$ tokens are processed simultaneously via CKKS SIMD packing.
    $\sf{VecGen}$, $\sf{PCMM}$, and $\sf{Total}$ report the amortized time (ms) for vector generation, plaintext-ciphertext matrix multiplication, and the overall $\sf PEL$ pipeline, respectively.
    $\sf{Depth}$ measures the multiplicative depth consumed during evaluation.
    IMDB accuracy indicates the embedding quality after compression; the uncompressed baseline accuracies are $0.8570$ for GloVe.6B.50d and $0.8857$ for GloVe.42B.300d. We do not report IMDB accuracy for GPT-2
    because less than $10\%$ of IMDB tokens overlap with its vocabulary.
    All timings are averaged over 10 runs, except Kim et al. with $\log p=10$, which is averaged over 2 runs due to its high computational cost.}
    \label{tab:compare}
    {
    \begin{tabular}{cccccccc}
    \toprule
    Dataset
    & $\log p$
    & Method 
    & \multicolumn{1}{c}{\makecell{\sf VecGen}}
    & \multicolumn{1}{c}{\makecell{\sf PCMM}} 
    & \multicolumn{1}{c}{\makecell{\sf Total}}
    & \multicolumn{1}{c}{\makecell{\sf Depth}}
    & \multicolumn{1}{c}{\makecell{IMDB Accuracy}}\\
    \cmidrule(lr){1-1}\cmidrule(lr){2-3}\cmidrule(lr){4-6}\cmidrule(lr){7-7}\cmidrule(lr){8-8}\\[-3.3ex]
    \cmidrule(lr){1-1}\cmidrule(lr){2-3}\cmidrule(lr){4-6}\cmidrule(lr){7-7}\cmidrule(lr){8-8}

    \multirow{6}{*}{\makecell{GloVe.6B.50d}}
    & \multirow{2}{*}{\makecell{6}} & Ours & 0.1953 & 0.1526 & 0.3480 & 6 & \multirow{2}{*}{\makecell[c]{0.7797}} \\
    &                             & Kim et al. & 11.7711 & 0.1729 & 11.9440 & 20 & \\    
    \cmidrule(lr){2-8}
    
    & \multirow{2}{*}{\makecell{8}} & Ours & 0.7895 & 0.5583 & 1.3478 & 8 & \multirow{2}{*}{\makecell[c]{0.8147}} \\
    &                            & Kim et al. & 69.4889 & 0.6201 & 70.1090 & 24 & \\    
    \cmidrule(lr){2-8}
    
    & \multirow{2}{*}{\makecell{10}} & Ours & 3.1543 & 2.2918 & 5.4461 & 10 & \multirow{2}{*}{\makecell[c]{0.8073}} \\
    &                            & Kim et al. & 424.7096 & 2.3099 & 427.0195 & 29 & \\
    \bottomrule

    \multirow{6}{*}{\makecell{GloVe.42B.300d}}
    & \multirow{2}{*}{\makecell{6}} & Ours & 0.1942 & 0.6723 & 0.8666 & 6 & \multirow{2}{*}{\makecell[c]{0.8053}} \\
    &                             & Kim et al. & 11.6995 & 0.6948 & 12.3943 & 20 & \\    
    \cmidrule(lr){2-8}
    
    & \multirow{2}{*}{\makecell{8}} & Ours & 0.8032 & 2.2757 & 3.078 & 8 & \multirow{2}{*}{\makecell[c]{0.8340}} \\
    &                            & Kim et al. & 70.6454 & 2.4227 & 73.0681 & 24 & \\    
    \cmidrule(lr){2-8}
    
    & \multirow{2}{*}{\makecell{10}} & Ours & 3.1488 & 8.7470 & 11.8958 & 10 & \multirow{2}{*}{\makecell[c]{0.8453}} \\
    &                            & Kim et al. & 423.5058 & 9.5222 & 433.0280 & 29 & \\
    \bottomrule
    
    \multirow{6}{*}{\makecell{GPT-2}}
    & \multirow{2}{*}{\makecell{6}} & Ours & 0.1932 & 1.6957 & 1.8889 & 6 & \multirow{2}{*}{\makecell[c]{-}} \\
    &                             & Kim et al. & 12.3247 & 1.7110 & 14.0357 & 20 & \\    
    \cmidrule(lr){2-8}
    
    & \multirow{2}{*}{\makecell{8}} & Ours & 0.7974 & 5.5967 & 6.3942 & 8 & \multirow{2}{*}{\makecell[c]{-}} \\
    &                            & Kim et al. & 70.9080 & 5.6021 & 76.5101 & 24 & \\    
    \cmidrule(lr){2-8}
    
    & \multirow{2}{*}{\makecell{10}} & Ours & 3.1385 & 20.8469 & 23.9855 & 10 & \multirow{2}{*}{\makecell[c]{-}} \\
    &                            & Kim et al. & 423.9253 & 20.8978 & 444.8231 & 29 & \\
    \bottomrule    
    \end{tabular}}
\end{table*}

\subsection{End-to-End Encrypted FastText Inference}\label{Section: Implementation2}
We next integrate private embedding lookup into an end-to-end encrypted text classification pipeline. This experiment evaluates whether the performance improvement observed in private embedding lookup leads to faster end-to-end inference while preserving the classification accuracy of the underlying plaintext model.

\subsubsection{FastText Model Structure}\label{Subsubsection:Fasttext}

\begin{figure}[h]
    \centering
    \includegraphics[width=1\linewidth]{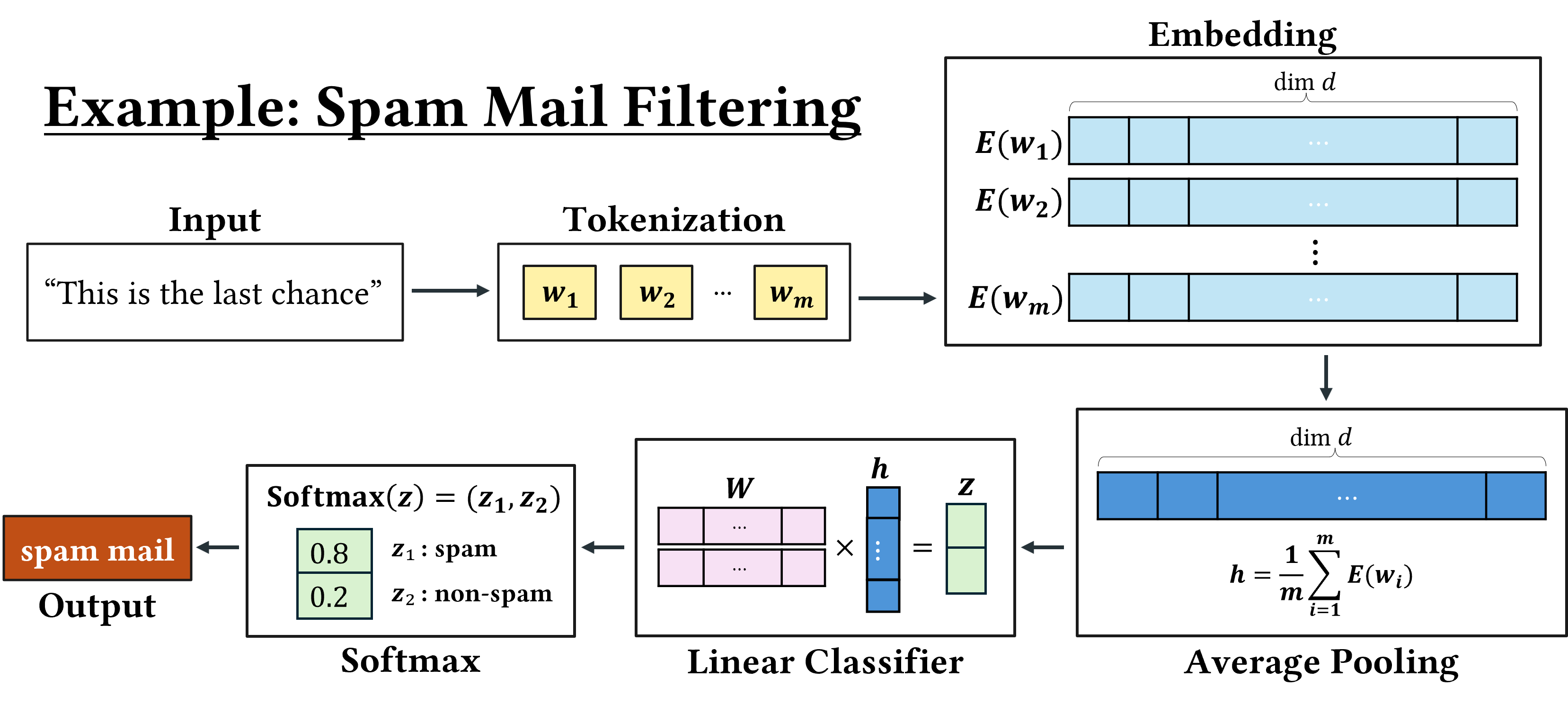}
    \caption{Overview of the fastText model used in our implementation. The model maps input tokens to embeddings, averages them into an input-level representation, and applies a single linear classifier without a bias term.}
    \label{fig:fasttext}
\end{figure}

FastText~\cite{fasttext} is a shallow text classification model showing that fast and accurate text classification can be achieved without a complex deep neural network, using only the average of word embeddings followed by a single linear classifier. As illustrated in Figure~\ref{fig:fasttext}, an input sentence is first tokenized into words, and each word is mapped to an embedding vector. The sentence representation is then obtained by averaging these word embeddings. Finally, a linear classifier is applied to compute the class scores.

In this work, we use the same FastText structure during training, but combine the embedding lookup and the linear classifier during encrypted inference. Specifically, the class score can be rewritten as
\[
z
=
W\left(\frac{1}{m}\sum_{i=1}^{m} E[w_i]\right)
=
\frac{1}{m}\sum_{i=1}^{m} W E[w_i].
\]
Therefore, the final class score is obtained by looking up the precomputed score vector \(WE[w_i]\) for each word and summing these vectors.

After decryption, the predicted label is determined by taking \(\arg\max\) over the resulting scores. The averaging factor \(1/m\) is a positive scalar applied equally to all class scores, and hence it does not affect the \(\arg\max\) result. Therefore, in the actual encrypted inference, we compute
\[
\tilde z = \sum_{i=1}^{m} W E[w_i]
\]
and determine the class label after decryption as
\[
\hat y = \arg\max_j \tilde z_j.
\]
Since \(\tilde z = mz\) and \(m>0\), this gives the same \(\arg\max\)-based prediction as the original FastText model.



\subsubsection{Implementation Detail}

\paragraph{Datasets}
We use two datasets for end-to-end encrypted FastText inference.
These datasets represent privacy-sensitive text domains: emails may contain personal or confidential communication, while drug reviews may reveal health-related information.
Here, a unigram means a single word token, and a bigram means an adjacent pair of words; for example, the sentence ``this is good'' contains the unigram tokens ``this'', ``is'', and ``good'', and the bigram tokens ``this is'' and ``is good''.

\begin{itemize}
    \item \textbf{Enron-Spam}~\cite{Enron}: email spam classification
    consisting of 33,716 emails with binary spam/ham labels.
    We concatenate the subject and body of each email and use unigram tokens.
    We set the embedding dimension to $d=50$.
    Existing plaintext baselines report around $97\%$--$99\%$
    accuracy~\cite{Enron_baseline1, Enron_baseline2, Enron_baseline3}.

    \item \textbf{Drugs.com Review}~\cite{Drug}: drug-review sentiment
    classification consisting of 215,063 reviews.
    We binarize the original rating into negative/positive sentiment labels
    and use both unigram and bigram tokens from the review text.
    We set the embedding dimension to $d=300$.
    Existing plaintext baselines report around $88\%$--$94\%$
    accuracy~\cite{Drug_baseline1, Drug_baseline2, Drug_baseline3}.
\end{itemize}

Detailed preprocessing rules, split construction, vocabulary construction are given in Appendix~\ref{Appendix:Experiment}.

\paragraph{Experimental Setup}
Unlike Section~\ref{Section: Implementation1}, which compresses existing
pretrained embedding tables, this experiment trains the compressed embedding
representation from scratch as part of the FastText classifier.
For both datasets, we use $\ell=4$ compressed subtables, each of size $p=256$,
so the total compressed embedding table size is
$\ell p \times d = 1024 \times d$, where $d$ denotes the embedding dimension.

During plaintext training, each input example is padded or truncated to $128$
tokens, and the compressed subtables, the mapping from each token to an
$\ell$-tuple of subtable indices, and the linear classifier are trained jointly
for the target classification task.
After training, each token is mapped to a fixed index tuple
$(j_1,\ldots,j_\ell)\in[p]^\ell$ with $\ell=4$.
This index tuple is encrypted and used as the private lookup query during
encrypted inference.
Since both tasks are binary classification tasks, the linear classifier is
represented by a matrix $W\in\mathbb{R}^{2\times d}$ and is trained without a
bias term.
Detailed plaintext training settings and hyperparameters are provided in
Appendix~\ref{Appendix:Experiment}.

As explained in Section~\ref{Subsubsection:Fasttext}, we fold the linear classifier into the lookup using
precomputed score subtables corresponding to $WE$.
Thus, encrypted inference consists of $\sf VecGen$ followed by $\sf PCMM$ over
the $\ell$ score subtables, and the resulting score contributions are summed
over the $\ell$ subtables and the $128$ tokens.

\paragraph{CKKS Parameter}
As shown in the Table~\ref{tab:compare}, our method requires a smaller multiplicative depth than Kim et al.'s method. This directly benefits end-to-end encrypted inference, where bootstrapping is often the dominant bottleneck. Reducing the required depth can lower the number of bootstrapping operations and may also enable lighter parameters, such as a smaller modulus budget or a lower ring degree, $N=2^{16}$.
In lightweight models such as FastText, where end-to-end encrypted inference can be performed without bootstrapping, these lighter parameters directly reduce the cost of the entire inference pipeline.
Thus, while Table~\ref{tab:compare} uses the same CKKS parameter set for both
methods, the end-to-end experiment uses separate parameter sets chosen to
satisfy the depth requirement of each method.
The concrete parameter sets are given in Table~\ref{tab:ckks_param}. Both sets provide at least 128-bit security based on the lattice estimator~\cite{Lattice_Est}.

\begin{table}[h]
\centering
\caption{CKKS parameter sets used for the implementation results.
For Table~\ref{tab:compare}, both methods use the same parameter set.
The terms $N$, $PQ$, $\Delta$, $\sf dnum$, and $h$ denote the ring degree, largest modulus, scale factor, gadget rank, and Hamming weight, respectively.}
\label{tab:ckks_param}
\begin{tabular}{c c c c c c c}
\hline
 & Method & $\log N$ & $\log PQ$ & $\log \Delta$ & $\sf dnum$ & $h$\\
\hline
Table~\ref{tab:compare} & - & 17 & 2070 & 51 & 3 & 64\\
\hline
\multirow{2}{*}{Table~\ref{tab:fasttext}}
& Ours & 16 & 1555 & 42 & 5 & 192 \\
& Kim et al. & 17 & 2070 & 51 & 3 & 64\\
\hline
\end{tabular}
\end{table}

\paragraph{CKKS SIMD Packing Setup}
For encrypted inference, each input example (an email in Enron-Spam or a review in Drugs.com Review) is represented by $128$ tokens, and each token is encoded by $\ell=4$ subtable indices.
In our implementation, the $\ell$ subtable-index streams are stored in $\ell$ separate ciphertexts during the $\sf VecGen$ stage.
Thus, each ciphertext uses $128$ SIMD slots per input example.
A CKKS ciphertext with ring dimension $N$ provides $N/2$ SIMD slots.
Therefore, our parameter set with $N=2^{16}$ provides $2^{15}$ slots and can pack $2^{15}/128=256$ examples, while Kim et al.'s parameter set with $N=2^{17}$ provides $2^{16}$ slots and can pack $2^{16}/128=512$ examples.
The time for each method is divided by the number of input examples processed simultaneously, and is reported as per-example time.

\subsubsection{Implementation Results}
Table~\ref{tab:fasttext} shows the end-to-end encrypted FastText inference results on the two text classification datasets.
Both datasets use the same compressed-table configuration with $\ell=4$ and $p=256$, but represent different practical text domains: email filtering and health-related review classification.
For each dataset, both methods use the same plaintext FastText model and therefore achieve the same classification accuracy.
The resulting accuracies are comparable to the plaintext baselines reported for the corresponding datasets.
The difference in encrypted inference comes from the private lookup method and the CKKS parameter set enabled by its multiplicative-depth requirement.

The main performance difference comes from the $\sf VecGen$ stage.
In the Enron-Spam experiment, $\sf VecGen$ accounts for $99.6\%$ of the total amortized inference time in Kim et al.'s method, making it the dominant part of the encrypted pipeline.
IVE-PEL reduces the amortized $\sf VecGen$ time from $9.0630\,$s to $0.0552\,$s per example, and the resulting $\sf VecGen$ fraction becomes $66.3\%$.
The $\sf Linear$ stage is also faster in our implementation, because our smaller depth requirement allows the use of a smaller ring degree. Together, these effects reduce the total amortized inference time from $9.1017\,$s to $0.0832\,$s, achieving a $109.4\times$ speedup. The Drugs.com Review experiment shows the same trend: IVE-PEL reduces the total amortized inference time from $9.1843\,$s to $0.0862\,$s, achieving a $106.5\times$ end-to-end speedup.

These results show that our method can be integrated into an end-to-end encrypted FastText inference pipeline while preserving baseline classification accuracy.
The same design principle also applies to NLP inference pipelines that start from an embedding layer and process many token queries.
In such settings, the proposed method can improve end-to-end encrypted inference by reducing the cost of server-side private lookup at the input stage.

\begin{table*}[t]
    \centering
    \caption{End-to-end encrypted FastText inference on two datasets.
    For all experiments, the FastText embedding layer is represented by four compressed subtables with $\log p=8$.
    Under their respective CKKS parameter sets, $256$ input examples are processed simultaneously by our method, while $512$ input examples are processed simultaneously by Kim et al.'s method.
$\sf{VecGen}$, $\sf{Linear}$, and $\sf{Total}$ report the amortized per-example time (s) for vector generation, the linear computation combining lookup and classification, and the overall encrypted inference pipeline, respectively.
Accuracy reports the plaintext test accuracy of the same FastText model evaluated in encrypted inference. All timings are averaged over 10 independent runs.}
    \label{tab:fasttext}
    {
    \begin{tabular}{ccccccc}
    \toprule
    Dataset
    & Method 
    & \multicolumn{1}{c}{\makecell{\sf VecGen}}
    & \multicolumn{1}{c}{\makecell{\sf Linear}} 
    & \multicolumn{1}{c}{\makecell{\sf Total}}
    & \multicolumn{1}{c}{\makecell{Accuracy}}\\
    \cmidrule(lr){1-1}\cmidrule(lr){2-2}\cmidrule(lr){3-5}\cmidrule(lr){6-6}\\[-3.3ex]
    \cmidrule(lr){1-1}\cmidrule(lr){2-2}\cmidrule(lr){3-5}\cmidrule(lr){6-6}

    \multirow{2}{*}{\makecell{Enron-Spam}}
    & Ours & 0.0552 & 0.0280 & 0.0832 & \multirow{2}{*}{\makecell[c]{0.9887}} \\
    & Kim et al. & 9.063 & 0.0386 & 9.1017 &\\    
    \cmidrule(lr){1-6}
    
    \multirow{2}{*}{\makecell{Drugs.com Review}}
    & Ours & 0.0568 & 0.0294 & 0.0862 & \multirow{2}{*}{\makecell[c]{0.9015}} \\
    & Kim et al. & 9.1451 & 0.0392 & 9.1843 &\\    
    \bottomrule
  
    \end{tabular}}
\end{table*}

\section{Conclusion}
In this work, we formalized \emph{Private Embedding Lookup} (PEL), a primitive that homomorphically performs the lookup of the corresponding embedding vector from an encrypted index query.
To make PEL practical, we consider both client-side communication and server-side computation.
Rather than constructing a one-hot vector and multiplying it by the table, we propose a new evaluation algorithm based on an orthogonal \emph{Discrete Cosine Transform} (DCT) basis, leveraging offline preprocessing of the embedding table and the complex-arithmetic features of CKKS.
Starting from encrypted queries that carry only single-slot index information, our design minimizes communication and reduces the dominant vector generation cost to $O(p)$, significantly lowering server-side computation.
In our implementation, this improves amortized lookup time by up to $78.4\times$ over the prior one-hot-based method.
In end-to-end encrypted FastText inference on the Enron-Spam dataset, it lowers the share of vector generation in the total inference time from $99.6\%$ to $66.3\%$, showing that the faster lookup also speeds up the overall encrypted inference.

\appendix

\begin{acks}
\end{acks}

\begin{ethics}
This work aims to improve the efficiency of privacy-preserving machine learning by reducing the cost of encrypted embedding lookup.
The proposed techniques are intended to help protect client inputs during outsourced inference and do not require the collection of new personal data or interaction with human subjects.
Our experiments use publicly available datasets and pretrained embedding tables, and we report only aggregate performance and accuracy results.

As with other cryptographic and privacy-enhancing techniques, deployment requires careful parameter selection and implementation review.
Incorrect parameters or side-channel vulnerabilities outside the formal protocol model could weaken the intended privacy guarantees.
We therefore view this work as a building block for privacy-preserving inference systems rather than a complete end-to-end security solution.
\end{ethics}

\begin{openscience}
Our implementation is currently built on the HEaaN library, which is not publicly available. Therefore, we are unable to release the full implementation artifact at this time. This limitation is due to the underlying HE library dependency rather than the datasets or training data used in our experiments. All datasets used in this work are publicly available, and the embedding-compression procedure follows the public implementation of~\cite{Github}. The training and preprocessing settings are described in Appendix~\ref{Appendix:Experiment}, so the non-HE parts of the experimental setup are based on public data and documented procedures.

Furthermore, the IVE algorithm and the overall PEL pipeline are described explicitly in the paper. Therefore, although we cannot release the current HEaaN-based implementation itself, the method can be implemented from the algorithmic description. 
\end{openscience}

\begin{ai}
The authors used AI-based writing and coding assistance during manuscript preparation.
The tools were used to assist with language polishing, LaTeX editing, consistency checks across sections, and drafting candidate revisions for explanatory text and proof presentation.
They were not used to generate experimental data, select baselines, validate cryptographic parameters, or make final scientific judgments.
All mathematical claims, proofs, experiments, references, and final manuscript text were reviewed and approved by the authors, who take full responsibility for the content.
\end{ai}

\bibliographystyle{ACM-Reference-Format}
\bibliography{example_paper}

\section{Details on CKKS}~\label{Appendix:CKKS}
The CKKS homomorphic encryption scheme~\cite{CKKS} provides approximate arithmetic on vectors of complex numbers. 

\subsection{CKKS Encoding} 
    
For a power-of-two $N$, let $\cc R_N := \bb Z[X]/(X^N+1)$ be the plaintext space. We call $N$ the \emph{ring degree}. If $N$ is clear from the context, we write $\cc R_N = \cc R$.
Note that $\cc R\otimes \bb R = \bb R[X]/(X^N+1)$ can be viewed as a $\bb C$-vector space of dimension $N/2$ by identifying $X^{N/2}$ with the imaginary unit $\sqrt{-1}$.

We consider two bijections between $\cc R\otimes \bb R$ and $\bb C^{N/2}$. 
The first is a $\bb C$-algebra isomorphism given by the Discrete Fourier Transform:
\[
{\sf DFT}:p(X)\mapsto \left(p(\zeta_i)\right)_{i\in [N/2]},
\]
with inverse ${\sf iDFT}$, where $\zeta_i=\exp\left(\frac{5^i\sqrt{-1}\pi}{N}\right)$ denotes the $2N$-th roots of unity.
    
The second is
\[
{\sf iCoeff}:\sum_{i\in[N]} m_iX^i \mapsto(m_i+\sqrt{-1}m_{i+N/2})_{i\in[N/2]},
\]
with inverse ${\sf Coeff}$. 
Note that ${\sf iCoeff}$ and ${\sf Coeff}$ are not $\bb C$-algebra isomorphisms, but only $\bb C$-vector space isomorphisms.\footnote{That is, a $\bb C$-algebra homomorphism preserves addition, scalar multiplication, and multiplication, whereas a $\bb C$-vector-space homomorphism need not preserve multiplication.}

Also note that for every $p(X)\in \cc R\otimes \bb R$ and $V=(\zeta_i^j)_{i,j}\in \bb C^{N/2\times N/2}$, we have
\[
{\sf DFT}(p(X)) = V\cdot {\sf iCoeff}(p(X)).
\]
Or equivalently, for every $\mathbf z\in \bb C^{N/2}$,
\[
{\sf iDFT}({\bf z})={\sf Coeff}(V^{-1}\cdot {\bf z}).
\]
A complex vector $\mathbf z\in \bb C^{N/2}$ is encoded into a plaintext $m(X)\in \cc R$ as follows.
An encoding map is parameterized by two components: a bijection ${\sf f}:\bb C^{N/2}\to \cc R\otimes\bb R$ and a scaling factor $\Delta$.
For a given map ${\sf f}:\bb C^{N/2}\to \cc R\otimes \bb R$---typically including ${\sf iDFT}$ and ${\sf Coeff}$, among many other possible choices---the vector $\mathbf z$ is first mapped to a polynomial ${\sf f}(\mathbf z)\in \cc R\otimes \bb R$, then scaled by $\Delta$ and rounded, yielding
\[
{\sf Ecd}(\mathbf z)=m(X)=\round{\Delta\cdot {\sf f}(\mathbf z)}\in \cc R.
\]
The resulting rounding error
\[
\delta_{\sf Ecd}(X)=\Delta\cdot {\sf f}(\mathbf z)-\round{\Delta\cdot {\sf f}(\mathbf z)}
\]
is called the \emph{encoding error}.

When ${\sf f}={\sf iDFT}$, $m$ is referred to as the \emph{slot encoding} of $\mathbf z$, and when ${\sf f}={\sf Coeff}$, $m$ is referred to as the \emph{coefficient encoding} of $\mathbf z$. Since ${\sf DFT}$ is an isomorphism that preserves all operations (addition, scalar multiplication, and multiplication), we adopt slot encoding as the standard encoding scheme. Unless stated otherwise, \emph{encoding} refers to slot encoding by default.

When decoding a polynomial $m(X)$, we rescale and apply ${\sf f}^{-1}$, yielding 
\[
{\sf Dcd}(m(X))={\sf f}^{-1}(\Delta^{-1} m(X))\in \bb C^{N/2}.
\]

\paragraph{Error Amplification and Encoding Error.}
Assume that ${\sf f}$ is $\bb R$-linear. To study error amplification, fix an encoding map and its corresponding decoding map.
Consider a polynomial $\Delta{\sf f}(\mathbf z)\in \cc R\otimes \bb R$ that decodes to $\mathbf z\in \bb C^{N/2}$.
Then, the decoding error induced by perturbing $\Delta{\sf f}(\mathbf z)$ to $\Delta{\sf f}(\mathbf z)+e(X)$ is
\[
\begin{aligned}
        &{\sf Dcd}(\Delta{\sf f}(\mathbf z)+e(X))-{\sf Dcd}(\Delta{\sf f}(\mathbf z))\\
    = \;&{\sf f}^{-1}\!\bigl(\Delta^{-1}(\Delta{\sf f}(\mathbf z)+e(X))\bigr)
       - {\sf f}^{-1}\!\bigl(\Delta^{-1}\Delta{\sf f}(\mathbf z)\bigr) \\
    = \;&\Delta^{-1}{\sf f}^{-1}(e(X)).
\end{aligned}
\]
Therefore, for a small error polynomial $e(X)\in \cc R\otimes \bb R$, its impact on the decoded complex vector is determined by the magnitude of $e(X)$, the amplification behavior of ${\sf f}^{-1}$, and the scaling factor $\Delta$.
When ${\sf f}^{-1}={\sf DFT}$, the error per vector component is heuristically amplified by a factor of about $\sqrt{N}$.
When ${\sf f}^{-1}={\sf iCoeff}$, no such amplification occurs, and the induced error is on the order of $\Delta^{-1}\|e(X)\|_{\infty}$.
    
The encoding error $\delta_{\sf Ecd}(X)$ also propagates through the same mechanism.
Here, the plaintext-domain error is $\delta_{\sf Ecd}(X)$ with $\|\delta_{\sf Ecd}(X)\|_{\infty}\leq0.5$.
If slot encoding is used, this error can be amplified by a factor of approximately $\sqrt{N}$.

\subsection{CKKS Encryption}

From now on, for notational convenience we omit the indeterminate $X$ and write polynomials without explicitly indicating it (e.g., $a(X)$ is written as $a$).
   
For an integer $Q$, define $\cc R_{N,Q} := \cc R_N / Q\cc R_N$.
When $N$ is clear from the context, we write $\cc R_Q$ instead of $\cc R_{N,Q}$; there should be no risk of confusion between $\cc R_N$ and $\cc R_Q$.
We call $Q$ a \emph{modulus}.
Let integers $q_i \approx \Delta$ for $i=1, \dots, L$ and $q_0\geq\Delta$ be given, and let the modulus $Q_\ell=\prod_{i=0}^\ell q_i$. We call $\{Q_\ell\}_{\ell\in[L]}$ a modulus chain.
Then a \emph{ciphertext} at level $\ell$ is a pair $(a,b)\in \cc R_{Q_\ell}^2$.

If $\|m\|_\infty < Q_0/2$, $m$ can be regarded as an element of $\cc R_{Q_\ell}$, and every element of $\cc R_{Q_\ell}$ can be uniquely lifted to an element of $\cc R$ with infinity norm at most $Q_0/2$.
Under the RLWE assumption~\cite{RLWE}, $m\in \cc R_{Q_\ell}$ can be encrypted into $(a,b)\in \cc R_{Q_\ell}^2$ satisfying 
\[
as+b=m+e_{\sf RLWE}\in \cc R_{Q_\ell}
\]
for some secret $s\in \cc R$ with $\|s\|_1=h$ and $\|s\|_{\infty}=1$ (Hamming weight $h$), and an RLWE error $e_{\sf RLWE}\gets \chi_{err}$. 

Depending on the context, we say that $(a,b)\in \cc R_{Q_{\ell}}^{2}$ is an encryption of $as+b\in \cc R_{Q_{\ell}}$, or equivalently, of $[as+b]_{Q_{\ell}}\in \cc R$; when ${\sf f}$ and $\Delta$ (and hence ${\sf Dcd}$) are specified, we may also say that it is an encryption of ${\sf Dcd}([as+b]_{Q_{\ell}})\in \bb C^{N/2}$.

\paragraph{Encryption Error.}
Here, the additional error introduced in the underlying plaintext is $e_{\sf RLWE}$, and the corresponding error in the decoded complex vector is given by $\Delta^{-1}{\sf f}^{-1}(e_{\sf RLWE})$; the same discussion as above applies.

\subsection{CKKS Operations} 

CKKS supports the following operations. For convenience, we assume that the encoding is slot encoding.
\begin{itemize}
    \item \textbf{Addition}: Homomorphically performs addition, either between two ciphertexts or between a ciphertext and a plaintext.
    \item \textbf{Multiplication}: Performs multiplication over $\cc R$, either between two ciphertexts or between a ciphertext and a plaintext.
    Multiplication consumes a level: multiplying a level-$\ell$ ciphertext with a level-$\ell'$ ciphertext yields a ciphertext at level $\min\{\ell,\ell'\}-1$.
    \item \textbf{Galois Automorphism}: Homomorphically applies a Galois automorphism $\sigma\in {\sf Gal}(\cc R)$. Under slot encoding, these automorphisms correspond to entrywise complex conjugation $(z_i)_{i\in [N/2]}\mapsto (\overline{z_i})_{i\in [N/2]}$, a rotation $(z_i)_{i\in [N/2]}\mapsto (z_{i+r})_{i\in [N/2]}$ (with indices taken modulo $N/2$), or their composition.
\end{itemize}

Except for addition, these operations typically introduce a small error $e\in \cc R\otimes \bb R$ in the plaintext domain whose norm is heuristically $\|e\|_{\infty}=O(\sqrt h)$. By error amplification, this induces an error of $\Delta^{-1}{\sf f}^{-1}(e)\in \bb C^{N/2}$ in the decoded complex vector.

\subsection{CKKS Bootstrapping}
As multiplications are repeatedly performed, the ciphertext eventually exhausts its modulus, making further rescaling impossible. CKKS bootstrapping restores the modulus and enables additional multiplications.
Suppose we have a low-level CKKS ciphertext encrypting a vector $\mathbf{z}$, and we wish to perform bootstrapping on it.  
The CKKS bootstrapping procedure consists of the following four main steps: ${\sf S2C}$, ${\sf ModRaise}$, ${\sf C2S}$, and ${\sf EvalMod}$.
    
\begin{itemize}
    \item ${\sf S2C}$: Homomorphically converts a slot-encoded CKKS ciphertext into a coefficient-encoded CKKS ciphertext while approximately preserving the underlying complex vector. 
        
    \item ${\sf ModRaise}$: For a coefficient-encoded CKKS ciphertext $(a,b)\in \cc R^2_{Q_0}$ such that $as+b=m\in \cc R_{Q_0}$, this step lifts it to $\cc R^2_{Q_L}$ so that
    \[
    as+b=m+Q_0I\in \cc R_{Q_L}
    \]
    for some $I\in \cc R_{Q_L}$. (Note that ${\sf iCoeff}(Q_0I)\in Q_0\bb Z[{\sqrt{-1}}]^{N/2}$.)      
        
    \item ${\sf C2S}$: Homomorphically converts a coefficient-encoded CKKS ciphertext into a slot-encoded CKKS ciphertext while approximately preserving the underlying complex vector. 

    \item ${\sf EvalMod}$: Homomorphically removes the ${\sf iCoeff}(Q_0 I)\in Q_0\bb Z[\sqrt{-1}]^{N/2}$ term by evaluating the ``Mod'' function 
    $z\mapsto [z]_{Q_0/\Delta}$ on $Re(z)$ and $Im(z)$ independently. 
    If only the real value is required, we apply the Mod function only to $Re(z)$.

\end{itemize}

Therefore, the bootstrapping procedure can be expressed as ${\sf EvalMod}\circ {\sf C2S}\circ {\sf ModRaise}\circ {\sf S2C}$.
If the level restored by ${\sf ModRaise}$ exceeds the total level consumed by ${\sf S2C}$, ${\sf C2S}$, and ${\sf EvalMod}$, then bootstrapping successfully achieves its purpose.

\paragraph{Functional Bootstrapping.}
The ${\sf EvalMod}$ step can be replaced by evaluating a function defined modulo $Q_0/\Delta$; in this case, a bootstrapping procedure that both refreshes the level and evaluates a target function is called \emph{functional bootstrapping}~\cite{LUT1,LUT2}. A representative example is
\[
\mathbf{x} \mapsto \exp(2\pi {\sqrt{-1}} \mathbf{x})
\]
for $\mathbf{x} \in [0,1]^{N/2}$. (In this case, we require $Q_0=\Delta$.) We call the resulting functional bootstrapping procedure ${\sf ExpBTS}$. 
This procedure is similar to the standard CKKS bootstrapping described above, except that the ${\sf EvalMod}$ step is replaced by ${\sf EvalExp}$.
Consequently, ${\sf ExpBTS}$ can be expressed as ${\sf EvalExp} \circ {\sf C2S} \circ {\sf ModRaise} \circ {\sf S2C}$.

\subsection{Accuracy in CKKS}
Encoding, encryption, and each homomorphic operation in CKKS introduce an error of magnitude $O(\sqrt{h})$ in the plaintext domain. Under slot encoding, this error is mapped to the complex vector with an additional factor of $O(\Delta^{-1}\sqrt{N})$.
Therefore, the error added to the complex vector per operation is approximately $\tfrac{1}{2}\log(hN)-\log\Delta+O(1)$ bits with high probability~\cite{CKKS}.
Accordingly, CKKS with scaling factor $\Delta$ supports a per-slot precision (in bits) of
\begin{align*}
    \log \Delta - \tfrac{1}{2}\log(hN) - O(1).
\end{align*}
Hence, a larger scaling factor $\Delta$ yields correspondingly higher-precision approximate computations.

\begin{remark}
The error should not be interpreted as something that necessarily accumulates monotonically; rather, this behavior reflects the intrinsic nature of approximate computation, not a peculiarity of CKKS itself.
The classical question of whether meaningful computations remain possible under finite precision has been studied extensively; for controlling numerical error in finite-precision arithmetic, we refer to~\cite{Higham2002ASNA}.
\end{remark}

\subsection{Ring Packing}~\label{app:A.ringpacking}
Ciphertexts of different ring degrees are compatible (via an appropriate embedding). For simplicity, we describe only the cases of ring degree $N/2$ and $N$.

\paragraph{Ring-degree embedding (noise-free).}
Consider the following ring embedding:
\[
\iota:\cc R_{N/2,Q}\to \cc R_{N,Q},\qquad 1\mapsto 1,\; X\mapsto X^2.
\]
This map is well-defined since $(X^{N/2}+1)\mapsto (X^2)^{N/2}+1=X^N+1$.
Moreover, since $\iota$ is a ring homomorphism, if $(a,b)\in \cc R_{N/2,Q}^2$ satisfies $as+b=m$, then
$(\iota(a),\iota(b))\in \cc R_{N,Q}^2$ satisfies the following with respect to the secret key $\iota(s)$:
\begin{align*}
\iota(a)\,\iota(s)+\iota(b) &= \iota(as+b) \\
&= \iota(m)\in \cc R_{N,Q}.
\end{align*}
Hence, increasing the ring degree is purely algebraic and does not introduce any additional approximation error.

\paragraph{Packing for real bootstrapping + \texorpdfstring{$\sf EvalExp$}{EvalExp}.}
We now assume that bootstrapping is performed at ring degree $N$.
A coefficient-encoded plaintext can be written as
\begin{align*}
\Big(\sum_{i\in[N/2]} m_iX^i\Big)+X^{N/2}\Big(\sum_{i\in[N/2]} m_{i+N/2}X^{i}\Big),
\end{align*}
and let the first term be $r(X)$ and the second summation be $i(X)$.
Real bootstrapping recovers only the coefficients of $r(X)$. (To also recover $i(X)$, an additional $\sf EvalMod$ is required (depending on the option), which increases the runtime.)
$\sf EvalExp$ applies $x\mapsto \exp(2\pi\sqrt{-1}x)$ only to the recovered coefficients $\{m_i\}_{i\in[N/2]}$ of $r(X)$.

Therefore, for ring packing, we first lift $m\in \cc R_{N/2,Q}$ to $\iota(m)\in \cc R_{N,Q}$, and then concentrate all information of $m$ into the first $N/2$ coefficients that are accessible to $\sf EvalExp$.
To do so, we compute:
\begin{align*}
&\left(\iota(a)-X^{N/2+1}\iota(a)\right)\iota(s)+\left(\iota(b)-X^{N/2+1}\iota(b)\right)\\
&=\;\iota(m)-X^{N/2+1}\iota(m)\in \cc R_{N,Q}.
\end{align*}
In the resulting plaintext, the first $N/2$ coefficients contain all coefficient information of $m$ (at even/odd positions).
Thus, applying real bootstrapping followed by $\sf EvalExp$ to this ciphertext yields the desired values
$\{\exp(2\pi\sqrt{-1}m_i)\}_{i\in[N/2]}$
in the $N/2$ slots of a ring-degree-$N$ ciphertext.
\begin{remark}
To compute at ring degree $N$, a secret key $s\in \cc R_{N/2}$ cannot be used directly as $\iota(s)$; instead, one must use a key sampled from the secret-key distribution over $\cc R_{N}$. This can be handled via standard key-switching.
\end{remark}

\section{Deferred Proof for DCT Matrix}~\label{Appendix:DeferredProofs}
In this appendix, we provide proofs of the deferred arguments.
\subsection[Proof for Orthogonality of D]{Proof for Orthogonality of $D$}
    \begin{lemma}~\label{lemmadeferredproofs}
    Let $p$ be even integer.
        Consider a matrix $D=(D_{k,\ell})_{k,\ell}\in \bb R^{p\times p}$:
        \begin{equation*}
        \resizebox{\columnwidth}{!}{$
        D_{k,\ell} = \begin{cases} \sqrt{\tfrac{2}{p}}\, \cos\!\left( \tfrac{(-1)^\ell(2\ell+1)\pi}{2p}\,(k+1) \right) & 0 \leq k < \tfrac{p}{2} \\[1.2ex] \sqrt{\tfrac{2}{p}}\, \sin\!\left( \tfrac{(-1)^\ell(2\ell+1)\pi}{2p}\,\bigl(k-\tfrac{p}{2}+1\bigr) \right) & \tfrac{p}{2} \leq k < p-1 \\[1.2ex] \tfrac{1}{\sqrt{p}} & k = p-1 \end{cases}
        $}
        \end{equation*}
       Then, $D$ is orthogonal.
    \end{lemma}

\begin{proof}
    It suffices to show $D^\top D=I.$
    Set $N:=p/2$ and define
\begin{equation*}
\omega_\ell := \frac{(-1)^\ell(2\ell+1)\pi}{2p},\qquad \ell=0,1,\dots,p-1.
\end{equation*}
For a fixed column index $\ell$, rewrite the entries of the $\ell$-th column using the
change of variables $r=k+1$ for $0\le k\le N-1$ and $r=k-N+1$ for $N\le k\le p-2$:
\begin{equation*}
D_{k,\ell} =
\begin{cases}
\sqrt{\tfrac{2}{p}}\cos(\omega_\ell r), & r=1,2,\dots,N,\\
\sqrt{\tfrac{2}{p}}\sin(\omega_\ell r), & r=1,2,\dots,N-1,\\
\tfrac{1}{\sqrt{p}}, & k=p-1.
\end{cases}
\end{equation*}
Hence, for any $\ell,\ell'$, with $\Delta:=\omega_\ell-\omega_{\ell'}$ and $\Sigma:=\omega_\ell+\omega_{\ell'}$, we obtain
\begin{equation*}
\resizebox{\columnwidth}{!}{$
\begin{aligned}
&(D^\top D)_{\ell,\ell'}
=\sum_{k=0}^{p-1} D_{k,\ell}D_{k,\ell'}\\
=\;\;&\frac{2}{p}\left(\sum_{r=1}^{N}\cos(\omega_\ell r)\cos(\omega_{\ell'} r)
+\sum_{r=1}^{N-1}\sin(\omega_\ell r)\sin(\omega_{\ell'} r)\right)+\frac{1}{p}\\
=\;\;&\frac{1}{p}\Bigl(1+\cos(N\Sigma)+\cos(N\Delta)
+2\sum_{r=1}^{N-1}\cos(r\Delta)\Bigr).
\end{aligned}
$}
\end{equation*}

\par\addvspace{\smallskipamount}
\noindent\textbf{Case 1: ($\ell=\ell'$).}
If $\ell=\ell'$, then $\Delta=0$ and $\Sigma=2\omega_\ell$, so
\begin{equation*}
\cos(N\Delta)=1,\qquad \sum_{r=1}^{N-1}\cos(r\Delta)=N-1.
\end{equation*}
Moreover, since $(2\ell+1)$ is odd,
\begin{equation*}
\cos(N\Sigma)=\cos(p\omega_\ell)
=\cos\!\left(\frac{(-1)^\ell(2\ell+1)\pi}{2}\right)=0.
\end{equation*}
Therefore
\begin{equation*}
(D^\top D)_{\ell,\ell}=\frac{1}{p}\bigl(1+0+1+2(N-1)\bigr)=\frac{2N}{p}=1.
\end{equation*}

\par\addvspace{\smallskipamount}
\noindent\textbf{Case 2: ($\ell\neq \ell'$).}
Define the odd integers
\begin{equation*}
q_\ell := (-1)^\ell(2\ell+1).
\end{equation*}
We first note that $q_\ell\equiv 1 \pmod 4$ for all $\ell$.
Hence $q_\ell-q_{\ell'}$ is divisible by $4$, so for some $m\in\mathbb{Z}\setminus\{0\},$
\begin{equation*}
\Delta=\omega_\ell-\omega_{\ell'}
=\frac{(q_\ell-q_{\ell'})\pi}{2p}
=\frac{2m\pi}{p}.
\end{equation*}
Consider the geometric sum
\begin{equation*}
S:=\sum_{r=-N+1}^{N}e^{{\sqrt{-1}}r\Delta}.
\end{equation*}
Since $2N=p$, we have $e^{{\sqrt{-1}}p\Delta}=e^{{\sqrt{-1}}2m\pi}=1$, while $\Delta\not\equiv 0\ (\mathrm{mod}\ 2\pi)$ for $\ell\neq \ell'$,
so
\begin{equation*}
\begin{aligned}
S
&=e^{-{\sqrt{-1}}(N-1)\Delta}\frac{1-e^{{\sqrt{-1}}p\Delta}}{1-e^{{\sqrt{-1}}\Delta}}
=0.
\end{aligned}
\end{equation*}
Taking real parts gives
\begin{equation*}
0=Re(S)=1+\cos(N\Delta)+2\sum_{r=1}^{N-1}\cos(r\Delta).
\end{equation*}
It remains to show $\cos(N\Sigma)=0$. Note that $q_\ell+q_{\ell'}\equiv 2\ (\mathrm{mod}\ 4)$, so $(q_\ell+q_{\ell'})/2$ is odd and
\begin{equation*}
\begin{aligned}
N\Sigma
&=\frac{p}{2}\cdot\frac{(q_\ell+q_{\ell'})\pi}{2p}
=\frac{(q_\ell+q_{\ell'})\pi}{4}
=\left(\text{odd}\right)\frac{\pi}{2},
\end{aligned}
\end{equation*}
hence $\cos(N\Sigma)=0$. Therefore, for $\ell\neq \ell'$,
\begin{equation*}
\begin{aligned}
&(D^\top D)_{\ell,\ell'}\\
=&\frac{1}{p}\Bigl(\underbrace{1+\cos(N\Delta)+2\sum_{r=1}^{N-1}\cos(r\Delta)}_{=\,0}
+\underbrace{\cos(N\Sigma)}_{=\,0}\Bigr)\\=&0.
\end{aligned}
\end{equation*}
Combining Steps 1 and 2 yields $D^\top D=I$.
\end{proof}

\subsection[Proof of v(alpha j)=v j]{Proof of ${\mathbf v}(\alpha_j)=\mathbf{v}_j$}
\begin{lemma}\label{lemmaB2}
Let $D\in \bb R^{p\times p}$ be given as in Lemma~\ref{lemmadeferredproofs}.
Let $\alpha_j := \exp\!\left(J\cdot2\pi\sqrt{-1} \right)$ for $J:=\tfrac{(-1)^j(2j+1)}{4p}$.
Let
\begin{equation*}\textstyle
{\bf v}(\alpha) := { \sqrt{\tfrac{2}{p}} 
\Big(Re(\alpha, \dots, \alpha^{\tfrac{p}{2}}), 
Im(\alpha, \dots, \alpha^{\tfrac{p}{2}-1}), 
\tfrac{1}{\sqrt{2}}\Big)^\top}.
\end{equation*}
Then, $j$th column of $D$ is equal to $\mathbf v(\alpha_j)$.
\end{lemma}
\begin{proof}
Fix $j\in\{0,1,\dots,p-1\}$ and define
\[
J:=\frac{(-1)^j(2j+1)}{4p},\qquad \alpha_j:=\exp\!\bigl(2\pi i J\bigr).
\]
Then
\begin{align*}
    \alpha_j=&\exp\!\left({\sqrt{-1}}\,\frac{(-1)^j(2j+1)\pi}{2p}\right)
=:e^{{\sqrt{-1}}\omega_j},\\
\omega_j:=&\frac{(-1)^j(2j+1)\pi}{2p},
\end{align*}
so that
\[
Re(\alpha_j^r)=\cos(r\omega_j),\;
Im(\alpha_j^r)=\sin(r\omega_j).
\]

We now compare the entries of the vector ${\bf v}(\alpha_j)$ in \eqref{eq:valpha}
with the $j$th column of $D$.

\par\addvspace{\smallskipamount}
\noindent\textbf{(i) Cosine block.}
For $0\le k<\tfrac{p}{2}$, set $r=k+1\in\{1,\dots,\tfrac{p}{2}\}$. Then the $k$th entry of
${\bf v}(\alpha_j)$ equals
\begin{equation*}\resizebox{\columnwidth}{!}{$
    \begin{aligned}
        {\bf v}(\alpha_j)_k=\sqrt{\frac{2}{p}}Re(\alpha_j^{r})
=\sqrt{\frac{2}{p}}\cos(\omega_j r)
=\sqrt{\frac{2}{p}}\cos\!\bigl(\omega_j (k+1)\bigr),
    \end{aligned}
    $}
\end{equation*}
which coincides with $D_{k,j}$ by the definition of $D$ (first case).

\par\addvspace{\smallskipamount}
\noindent\textbf{(ii) Sine block.}
For $\tfrac{p}{2}\le k< p-1$, set $r=k-\tfrac{p}{2}+1\in\{1,\dots,\tfrac{p}{2}-1\}$.  
Then the $k$th entry of ${\bf v}(\alpha_j)$ equals
\begin{equation*}\resizebox{\columnwidth}{!}{$
    \begin{aligned}
{\bf v}(\alpha_j)_k=\sqrt{\frac{2}{p}}Im(\alpha_j^{r})
=\sqrt{\frac{2}{p}}\sin(\omega_j r)
=\sqrt{\frac{2}{p}}\sin\!\Bigl(\omega_j\bigl(k-\tfrac{p}{2}+1\bigr)\Bigr),
    \end{aligned}
    $}
\end{equation*}
which coincides with $D_{k,j}$ by the definition of $D$ (second case).

\par\addvspace{\smallskipamount}
\noindent\textbf{(iii) Last entry.}
Finally, the last component of ${\bf v}(\alpha_j)$ is
\[
{\bf v}(\alpha_j)_{p-1}
=\sqrt{\frac{2}{p}}\cdot\frac{1}{\sqrt{2}}
=\frac{1}{\sqrt{p}}
=D_{p-1,j},
\]
matching the third case in the definition of $D$.
Since all entries agree, the $j$th column of $D$ is exactly ${\bf v}(\alpha_j)$.
\end{proof}

    

\section{Numerical Analysis}~\label{Appendix:MainThm}
We separate the total deviation into the error due to approximate evaluation and the error induced by input perturbation.
Let $f$ be the target function and let $\tilde f$ be a circuit (or procedure) that computes an approximation of $f$.
Given an input perturbation from $x$ to $\tilde x$, we write
\begin{equation}\label{eq:error-decomp}
\tilde f(\tilde x)-f(x)
=
\underbrace{\tilde f(\tilde x)-f(\tilde x)}_{\text{forward error}}
+
\underbrace{f(\tilde x)-f(x)}_{\text{backward error}}.
\end{equation}
The first term captures the approximation error incurred even on the same input $\tilde x$, whereas the second term captures the sensitivity of $f$ to the input perturbation.

When $f$ is differentiable, the backward error admits the first-order expansion
\begin{equation}\label{eq:taylor}
f(\tilde x)-f(x)=Df(x)\cdot(\tilde x-x)+o(\|\tilde x-x\|).
\end{equation}
Equivalently, if $f$ is locally Lipschitz around $x$ with constant $L$, then $|f(\tilde x)-f(x)|\le L\|\tilde x-x\|$.

We now instantiate this decomposition for CKKS.
CKKS supports fixed-point arithmetic via a scaling factor $\Delta$.
The per-slot precision is roughly $\log \Delta - \frac{1}{2}\log(hN)-O(1)$ bits; correspondingly, each basic arithmetic step (notably ciphertext multiplication followed by rescaling) introduces an absolute error of such order in the decoded (slot) value.
While the exact constant depends on parameters and implementation details, this model provides a convenient conservative abstraction for tracking accumulated approximation error along a circuit.
In particular, a coefficient-domain Gaussian error with coefficient standard deviation $\sigma$ is measured in the slot domain through its canonical embedding.
Following the standard CKKS noise analysis~\cite{CKKS}, we use the high-probability bound $6\sigma\sqrt N$ on this canonical-embedding norm.

With this convention, we proceed to the proof of the theorem.

\begin{theorem}\label{thm:app-admissible-params}
Fix a correctness parameter $u$, an even index size $p$ satisfying the assumptions of Section~\ref{IVE} (or its padded size), an embedding dimension $d$, and a CKKS ring degree $N$.
Let $\Delta$ be the CKKS scaling factor.
Let $C_{\sf PCMM}$ satisfy $\|\epsilon_{\sf PCMM}\|_{\sf RMS}\le C_{\sf PCMM}\Delta^{-1}$ for the final PCMM step, and let $C_T$ satisfy $\|M\|_2\le C_T(\sqrt d+\sqrt p)$ for tables $M\in\mathcal T$.
Let $C_V$ be a local Lipschitz constant satisfying $\|{\bf v}(\alpha+\epsilon)-{\bf v}(\alpha)\|_2\le C_V\frac{p+2}{\sqrt{12}}|\epsilon|$ in the perturbation range considered below.
Let $C_{\rm coeff}=19.2$ denote the coefficient-domain 6-sigma bound for Gaussian error of standard deviation $\sigma=3.2$, and set $C_0:=C_{\rm coeff}\sqrt N$ for the corresponding canonical-embedding bound.
Assume that after ring packing, the value entering IVE is $\tilde\alpha_i=\alpha_{j_i}+\epsilon_{\alpha,i}$ with
\[
\left|\epsilon_{\alpha,i}\right|
\le
\frac{C_0}{\Delta}.
\]
If
\begin{equation}\label{eq:app-sufficient-condition}
\frac{C_{\sf PCMM}}{\Delta}
+
\frac{C_T(\sqrt d+\sqrt p)}{\sqrt d}
\cdot
C_V\frac{p+2}{\sqrt{12}}
\cdot
\frac{C_0}{\Delta}
\le
\frac{1}{u},
\end{equation}
then \textsf{IVE-PEL} returns encryptions of $M_{j_i}$, for every packed index $j_i$, with RMS error at most $1/u$, with overwhelming probability over the protocol randomness under the stated CKKS error bounds.
\end{theorem}

\begin{proof}
We track one packed index $j_i$ and omit the subscript $i$.
The same argument applies independently to every packed slot.
All bounds below are conditioned on the high-probability CKKS error events stated in the theorem.
All tildes denote values computed from approximate arithmetic and/or perturbed inputs.

\begin{enumerate}[leftmargin=*, itemsep=-2pt, topsep=0pt]
    \item The server reconstructs $(a,b)$, applies ring packing, and obtains a ciphertext $\ct\in\cc R_{N,Q_0}^2$ whose decoded IVE input is $\tilde\alpha=\alpha+\epsilon_\alpha$.
    \item The server runs IVE to obtain encryptions of $\sqrt{2p}\,\tilde{\bf v}_j$.
    After the following multiplication by $ML/\sqrt{2p}$, we write $\epsilon_{{\bf v}_j}:=\tilde{\bf v}_j-{\bf v}_j$ for the corresponding unscaled vector error.
    \item The server runs ${\sf PCMM}_{ML/\sqrt{2p}}$ and obtains encryptions of $\tilde M_j$.
\end{enumerate}

\subsection*{Step 1: Bounding $\epsilon_{M_j}$ from $\epsilon_{{\bf v}_j}$}
Write the PCMM forward error as
\[
\epsilon_{\sf PCMM}:=\tilde M_j-ML\tilde{\bf v}_j,
\qquad
\|\epsilon_{\sf PCMM}\|_{\sf RMS}\le \frac{C_{\sf PCMM}}{\Delta}.
\]
Since $L$ is orthogonal, $\|ML\|_2=\|M\|_2$.
Using the table bound $\|M\|_2\le C_T(\sqrt d+\sqrt p)$, we get
\begin{align}
\|\tilde M_j-M_j\|_{\sf RMS}
&\le
\|\epsilon_{\sf PCMM}\|_{\sf RMS}
+
\frac{1}{\sqrt d}\|ML(\tilde{\bf v}_j-{\bf v}_j)\|_2 \nonumber\\
&\le
\frac{C_{\sf PCMM}}{\Delta}
+
\frac{C_T(\sqrt d+\sqrt p)}{\sqrt d}
\|\epsilon_{{\bf v}_j}\|_2 .
\label{eq:app-m-from-v}
\end{align}

\subsection*{Step 2: Bounding $\epsilon_{{\bf v}_j}$ from $\epsilon_{\alpha}$}
We use the following local Lipschitz bound for the IVE vector map.

\begin{lemma}\label{lem:v-deriv}
Let ${\bf v}(\alpha)$ be defined by
\[
{\bf v}(\alpha)
:= \sqrt{\tfrac{2}{p}}
\Big(Re(\alpha,\dots,\alpha^{p/2}),
Im(\alpha,\dots,\alpha^{p/2-1}),
\tfrac{1}{\sqrt2}\Big)^\top.
\]
Suppose $|\alpha|=1$ and $\epsilon$ lies in the perturbation range in which
\[
|(\alpha+\epsilon)^k-\alpha^k|\le C_V k|\epsilon|
\qquad(1\le k\le p/2)
\]
for an absolute constant $C_V$.
Then
\[
\|{\bf v}(\alpha+\epsilon)-{\bf v}(\alpha)\|_2
\le
C_V\frac{p+2}{\sqrt{12}}|\epsilon|.
\]
\end{lemma}

\begin{proof}[Proof of Lemma~\ref{lem:v-deriv}]
For each $1\le k\le p/2-1$, the real and imaginary coordinates together contribute at most
$|(\alpha+\epsilon)^k-\alpha^k|^2$.
The last real coordinate for $k=p/2$ contributes at most the same complex-power difference.
Therefore,
\begin{align*}
\|{\bf v}(\alpha+\epsilon)-{\bf v}(\alpha)\|_2^2
&\le
\frac{2}{p}C_V^2|\epsilon|^2\sum_{k=1}^{p/2}k^2\\
&=
C_V^2|\epsilon|^2\frac{(p+1)(p+2)}{12}\\
&\le
C_V^2|\epsilon|^2\frac{(p+2)^2}{12}.
\end{align*}
Taking square roots proves the claim.
\end{proof}

By Lemma~\ref{lem:v-deriv}, with $\epsilon_\alpha:=\tilde\alpha-\alpha$,
\begin{equation}\label{eq:app-v-from-alpha}
\|\epsilon_{{\bf v}_j}\|_2
\le
C_V\frac{p+2}{\sqrt{12}}|\epsilon_\alpha|.
\end{equation}

\subsection*{Step 3: Bounding the initial $\epsilon_{\alpha}$}
By the theorem assumption, the packed input satisfies
\[
\tilde\alpha=\alpha+\epsilon_\alpha,
\qquad
|\epsilon_\alpha|
\le
\frac{C_0}{\Delta}.
\]

Combining~\eqref{eq:app-m-from-v} and \eqref{eq:app-v-from-alpha} with this initial bound, we get
\[
\|\tilde M_j-M_j\|_{\sf RMS}
\le
\frac{C_{\sf PCMM}}{\Delta}
+
\frac{C_T(\sqrt d+\sqrt p)}{\sqrt d}
\cdot
C_V\frac{p+2}{\sqrt{12}}
\cdot
\frac{C_0}{\Delta}.
\]
The sufficient condition~\eqref{eq:app-sufficient-condition} makes the right-hand side at most $1/u$.
\end{proof}

The theorem gives the following parameter consequences.
Define
\[
S(\Delta)
:=
\left(\frac{\Delta}{u}-C_{\sf PCMM}\right)
\frac{\sqrt d}{C_V C_T(\sqrt d+\sqrt p)}
\frac{\sqrt{12}}{p+2}.
\]
For $\Delta/u>C_{\sf PCMM}$, condition~\eqref{eq:app-sufficient-condition} is equivalent to
\begin{equation}\label{eq:main-cond-S}
C_0\le S(\Delta).
\end{equation}

\begin{corollary}[Scale selection]
It suffices to choose
\begin{equation}\label{eq:Delta-min}
\Delta
\ge
uC_{\sf PCMM}
+
u\cdot
\frac{C_V C_T(\sqrt d+\sqrt p)}{\sqrt d}\cdot
\frac{p+2}{\sqrt{12}}\cdot
C_0.
\end{equation}
This is exactly condition~\eqref{eq:main-cond-S} written as a lower bound on $\Delta$.
\end{corollary}

\begin{corollary}[Asymptotic parameter relation]\label{THMCOR}
Assume $C_{\sf PCMM}=o(\Delta/u)$ and choose $\Delta$ on the scale of~\eqref{eq:Delta-min}.
Then
\begin{equation}\label{eq:logDelta-scale}
\begin{aligned}
\log\Delta
&=
\log\!\Bigl(C_V C_T\cdot u\cdot C_0\cdot(p+2)\cdot(1+\sqrt{p/d})\Bigr)+O(1)\\
&=
\log\!\Bigl(u\cdot C_0\cdot p\cdot(1+\sqrt{p/d})\Bigr)+O(1),
\end{aligned}
\end{equation}
where the final $O(1)$ absorbs absolute constants and table-dependent constants.
\end{corollary}

\begin{corollary}[Query-size accounting]
At modulus $Q_0=\Delta$, the query consists of a 128-bit seed and the $b$ component.
Sending one encrypted scalar query slot for each of the $N/2$ indices costs
\[
128+\frac{N}{2}\log Q_0
\]
bits.
Therefore the amortized multiplicative overhead relative to the $\log p$-bit plaintext index representation is
\[
\frac{128+\frac{N}{2}\log Q_0}{\frac{N}{2}\log p}
=
\frac{\log Q_0}{\log p}
+
\frac{256}{N\log p}.
\]
With $Q_0=\Delta$ and the asymptotic relation~\eqref{eq:logDelta-scale}, this becomes
\[
1+\frac{\log\!\bigl(uC_0(1+\sqrt{p/d})\bigr)}{\log p}
+O\!\left(\frac{1}{\log p}\right)
+\frac{256}{N\log p}.
\]
\end{corollary}

\section{Training Details}~\label{Appendix:Experiment}

In this appendix, we provide additional preprocessing and training details used in our
implementation.

\subsection{Implementation 1. Comparison of PEL Method}
There are two separate training stages in the experiments of Section~\ref{Section: Implementation1}.
First, we train compressed embedding tables from the source embeddings (GloVe.42B.300d, GloVe.6B.50d, GPT-2).
Second, we train an IMDB sentiment classifier using the compressed embeddings obtained in the first stage to evaluate downstream utility.


\paragraph{Compressed Embedding Training}
The table-compression structure is described in
Section~\ref{Section: Implementation1}.
Here we give the training details.
For each source embedding table, we train the compressed subtables using the
implementation of~\cite{Github}.
The training objective is the mean-squared error between the original embedding
and its reconstructed compressed embedding.
We use Adam with learning rate $10^{-4}$, batch size $128$, and train for
$20$ epochs.
We save the checkpoint with the lowest validation MSE and use this checkpoint
to generate the compressed subtables and the learned index tuples.
In all cases, we fix the number of subtables to $\ell=4$ and vary the subtable
size $p$.
Thus, the total compressed table size is $\ell p \times d = 4p \times d$.



\paragraph{IMDB Utility Evaluation}
To evaluate whether the compressed embeddings retain downstream utility, we
train a lightweight LSTM classifier on the IMDB sentiment classification
task~\cite{IMDB}.
For GloVe-based experiments, the compressed GloVe embeddings are used as
pretrained word embeddings for the classifier.
The classifier consists of a single LSTM layer with hidden dimension $150$,
followed by a linear classification layer.
We train the classifier with Adam, learning rate $10^{-4}$, batch size $128$,
and negative log-likelihood loss.
For GloVe.42B.300d, we train the classifier for $100$ epochs.
For GloVe.6B.50d, we use the same learning rate and train the classifier for
$300$ epochs.

For the IMDB preprocessing, we use the Keras IMDB dataset with maximum sequence
length $400$.
We build the classifier vocabulary from the intersection of the IMDB word
vocabulary and the corresponding GloVe vocabulary.
Words not covered by the GloVe vocabulary are mapped to \texttt{<unk>}, and
shorter sequences are padded to length $400$.
The IMDB training split is further divided into training and validation subsets:
we randomly sample $3$ batches of training examples for validation and use the
remaining examples for training.
We save the checkpoint with the lowest validation loss and report the test
accuracy of this checkpoint.
All reported IMDB accuracies are averaged over three runs.

We do not report IMDB accuracy for GPT-2 embeddings.
In our preprocessing pipeline, only about $5.4\%$ of IMDB words overlap with the
GPT-2 token vocabulary.
Consequently, more than $90\%$ of IMDB words are mapped to the unknown token,
making it difficult for the classifier to learn meaningful representations from
GPT-2 in this setting.




\subsection{Implementation 2. End-to-End FastText Inference}

\paragraph{Dataset Processing}
All datasets share a common tokenization pipeline.
We lowercase the text and extract alphabetic token sequences using the regular
expression, discarding digits, punctuation, and non-ASCII
symbols.
The vocabulary is constructed only from the training split to avoid label
leakage, and out-of-vocabulary tokens are mapped to a trainable \texttt{<unk>}
embedding.
Each input example is padded or truncated to $128$ tokens.

For Enron-Spam, we concatenate the email subject and message body into a single
document and binarize the label into spam and ham.
We split the 33,716 emails into training, validation, and test sets using a
stratified 80/10/10 split, resulting in 26,972 training samples, 3,372
validation samples, and 3,372 test samples.
We use unigram tokenization.

For Drugs.com Review, we use the review text and first apply HTML-entity
unescaping and whitespace normalization before tokenization.
We convert the original 1--10 rating into a binary sentiment label, where
ratings 1--5 are labeled negative and ratings 6--10 are labeled positive.
Starting from the official split of 161,297 training reviews and 53,766 test
reviews, we reserve 10\% of the official training set for validation, resulting
in 145,167 training samples, 16,130 validation samples, and 53,766 test samples.
We use unigram+bigram tokenization.
The vocabulary is built from the training split, and tokens appearing fewer than
three times are mapped to \texttt{<unk>}.

\paragraph{Plaintext Training Details}
In Implementation~2, we train a FastText-style classifier directly with
compressed embedding tables.
Unlike Implementation~1, where compressed tables are obtained by compressing
pretrained embeddings, the compressed subtables in this experiment are trained
from scratch for the target classification task.
For both datasets, we use $\ell=4$ compressed subtables, each of size $p=256$.
The compressed subtables, the mapping from each token to an $\ell$-tuple of
subtable indices, and the linear classifier are trained jointly in plaintext.

We train the model using cross-entropy loss and AdamW.
The mapping from tokens to subtable indices is discrete, but discrete index
selection is not directly differentiable.
Therefore, during training, we use a softmax relaxation: for each token and each
subtable, the model learns a probability distribution over the $p$ entries and
uses the corresponding weighted average of subtable entries.
After training, this soft assignment is converted into a hard assignment by
taking the argmax entry in each subtable.
Thus, each token is finally represented by an index tuple
$(j_1,\ldots,j_\ell)\in[p]^\ell$, which is used for encrypted inference.
The linear classifier is trained without a bias term.

For Enron-Spam, we set the maximum number of tokens per email to $128$, the
embedding dimension to $d=50$, the learning rate to $10^{-2}$, and train for
$20$ epochs.
For Drugs.com Review, we set the maximum number of tokens per review to $128$,
the embedding dimension to $d=300$, the learning rate to $7\cdot 10^{-3}$, and
train for $20$ epochs.
For Drugs.com Review, we also use code-sharpening regularization with entropy
weight $2\cdot 10^{-3}$ and confidence weight $10^{-2}$.

\section{Embedding Lookup in Smaller Queries}~\label{Appendix:SmallQueries}

In the main implementation results, we evaluate private embedding lookup in a
full-batch setting that maximizes the use of CKKS SIMD packing. Here, a batch
means the number of private token queries processed together in one ciphertext.
This setting captures a natural high-throughput regime where many token queries
are available at once, for example when a long document is processed, multiple
documents are evaluated together, or several inference requests are batched.

For a broader evaluation, we additionally consider smaller-batch settings, where
the number of packed queries is below the full SIMD capacity. This experiment
shows the performance of the proposed lookup procedure under different degrees
of SIMD utilization. We report these additional results in this appendix.

Unlike the main experiments, which use plaintext-ciphertext matrix-matrix
multiplication (PCMM) to process many generated vectors simultaneously, this
experiment implements the linear step using plaintext-ciphertext matrix-vector
multiplication. We implement this matrix-vector multiplication using the
baby-step giant-step (BSGS) method together with one-hoisting, following the
standard approach for homomorphic linear transformations~\cite{MatVecBSGS}.
This reduces the number of ciphertext rotations compared with the naive
diagonal method, while hoisting amortizes the common key-switching
decomposition across multiple rotations of the same ciphertext.

\paragraph{Packing Method}
To evaluate the smaller-batch regime, we instantiate the lookup task using the
GloVe.42B.300d embedding table.
As in the main implementation, each token is represented by four codebook
indices, and the lookup is performed over four corresponding subtables.
Since the embedding layer produces the first input representation for the
subsequent model computation, the output of private lookup should be packed in a
form that can be directly consumed by the downstream encrypted operations.
For this purpose, we use a strided sparse packing layout.

Concretely, each embedding vector has dimension $300$, and we pad it to
$2^{\lceil \log_2 300\rceil}=512$ coordinates.
Since one CKKS ciphertext contains $2^{16}$ slots, this packing layout supports
$2^{16}/512=128$ queried tokens in one ciphertext.
The resulting strided layout is given by
\[
\mathrm{slot}[128j+b]=
\begin{cases}
\mathbf{v}_b[j], & 0\le j<300,\; 0\le b<128,\\
0, & 300\le j<512,\; 0\le b<128.
\end{cases}
\]
Thus, the effective batch size in this smaller-batch setting is $128$.

When $p>512$, a selection vector cannot be represented within the slots allocated
to a single queried token in one ciphertext.
Therefore, two ciphertexts are required to represent the selection vector.
The case $\log p=10$ in our experiment corresponds exactly to this setting.
For our method, after power generation, the additional cost is limited because
the real and imaginary parts can be placed in separate ciphertexts.
In contrast, Kim et al.'s method must apply the same homomorphic operations to
the additional ciphertext as well.
Consequently, when $\log p=10$, the main evaluation cost of Kim et al.'s method
is incurred almost twice, whereas the additional cost of our method remains
relatively small.

\paragraph{Implementation Results}
Table~\ref{tab:small-batch} reports the smaller-batch evaluation results on GloVe.42B.300d. For all parameter settings, our method achieves faster evaluation than Kim et al.'s method when generating an output ciphertext containing $128$ embedding vectors. Specifically, our method reduces the total latency by $2.62\times$, $2.30\times$, and $2.83\times$ for $\log p=6,8,10$, respectively.

The speedup is less pronounced than in the full-batch setting
because Kim et al.'s method applies the same function independently to many entries of the
selection vector, and this element-wise structure benefits directly from CKKS
SIMD parallelism.
Nevertheless, our method still achieves a consistent $2$--$3\times$ improvement
in total latency.
This improvement mainly comes from the lower multiplicative depth of our
evaluation.
When a selection vector is processed within the same ciphertext, both our method
and Kim et al.'s method require $O(\log p)$ non-scalar multiplications
asymptotically.
However, Kim et al.'s method incurs a depth of $\log p+O(1)$, where the additional
constant-depth terms have a non-negligible impact in this parameter range.
In contrast, our method requires only $\log p$ multiplicative depth.
Therefore, when the output ciphertexts are matched at the same output modulus,
our method can start from a smaller initial modulus chain and perform the
evaluation with lighter ciphertext arithmetic.
This explains why the proposed method remains faster even in the smaller-batch
setting, where the SIMD advantage of the baseline is more directly exploited.

\begin{table}[h]
    \centering
    \caption{Smaller-batch evaluation using the same GloVe.42B.300d dataset as
    Table~\ref{tab:compare}. We use the same CKKS
    parameters as in Table~\ref{tab:compare}. Here, batch size denotes the number
    of embedding vectors packed in a single output ciphertext, where each vector
    corresponds to one private token query. For each batch size, we report the
    latency (s) required to generate one such output ciphertext. The linear step,
    denoted by MatVec, is implemented using plaintext-ciphertext matrix-vector
    multiplication rather than PCMM. Each latency is averaged over $10$ independent runs.}
    \label{tab:small-batch}
    {
    \begin{tabular}{ccccccc}
    \toprule
    $\log p$ 
    & Method
    & \multicolumn{1}{c}{\makecell{\sf VecGen}}
    & \multicolumn{1}{c}{\makecell{\sf MatVec}} 
    & \multicolumn{1}{c}{\makecell{\sf Total}}
    & \multicolumn{1}{c}{\makecell{Batch Size}}\\
    \cmidrule(lr){1-1}\cmidrule(lr){2-2}\cmidrule(lr){3-5}\cmidrule(lr){6-6}\\[-3.3ex]
    \cmidrule(lr){1-1}\cmidrule(lr){2-2}\cmidrule(lr){3-5}\cmidrule(lr){6-6}

    \multirow{2}{*}{\makecell{6}}
    & Ours & 6.2334 & 2.8236 & 9.0572 & \multirow{2}{*}{\makecell[c]{128}} \\
    & Kim et al. & 20.7202 & 3.0494 & 23.7695 &\\    
    \cmidrule(lr){1-6}
    
    \multirow{2}{*}{\makecell{8}}
    & Ours & 9.2418 & 5.5929 & 14.8345 & \multirow{2}{*}{\makecell[c]{128}} \\
    & Kim et al. & 28.6495 & 5.4167 & 34.0663 &\\    
    \cmidrule(lr){1-6}
    
    \multirow{2}{*}{\makecell{10}}
    & Ours & 13.5587 & 14.6284 & 28.1873 & \multirow{2}{*}{\makecell[c]{128}} \\
    & Kim et al. & 64.8553 & 14.9204 & 79.7758 &\\    
    \bottomrule
  
    \end{tabular}
    }
\end{table}

\end{document}